\RequirePackage{fix-cm}
\documentclass[smallextended]{svjour3}       
\smartqed  
\usepackage{graphicx}
\usepackage{multicol}
\usepackage{multirow}
\usepackage{comment}
\usepackage{amsmath}

\usepackage{amsthm}
\usepackage{dirtytalk}
\usepackage[10pt]{extsizes}
\usepackage{makecell}
\usepackage{soul}
\usepackage{color}
\usepackage{float}
\usepackage{subfig}
\usepackage{setspace}
\usepackage{lineno}
\usepackage{amssymb}
\usepackage{placeins}
\newtheorem{pro}{Proposition}
\newtheorem{dfn}{Definition}
\linenumbers 
\newcommand{\orcid}[1]{\href{https://orcid.org/#1}{\textcolor[HTML]{A6CE39}{\aiOrcid}}}
\newcommand\correspondingauthor{\thanks{*Corresponding author.}}

\DeclareMathOperator*{\argmin}{arg\,min}
\begin{document}
\nolinenumbers
\title{Footprints of Data in a Classifier: Understanding the Privacy Risks and Solution Strategies} 


\author{Payel Sadhukhan*\correspondingauthor       \and
        Tanujit Chakraborty 
}


\institute{Payel Sadhukhan \at
              Techno Main Salt Lake, Kolkata, India\\
              \email{payel0410@gmail.com}   \\        
              ORCiD: 0000-0001-7795-3385
           \and
           Tanujit Chakraborty \at
             Sorbonne University, Abu Dhabi, UAE\\
             Sorbonne Center for AI, Paris, France\\
              \email{tanujit.chakraborty@sorbonne.ae}   \\
              ORCiD: 0000-0002-3479-2187
}

\date{Received: date / Accepted: date}

\maketitle

\begin{abstract}
The widespread deployment of Artificial Intelligence (AI) across government and private industries brings both advancements and heightened privacy and security concerns. Article 17 of the General Data Protection Regulation (GDPR) mandates the Right to Erasure, requiring data to be permanently removed from a system to prevent potential compromise. While existing research primarily focuses on erasing sensitive data attributes, several passive data compromise mechanisms remain underexplored and unaddressed. One such issue arises from the residual footprints of training data embedded within predictive models. Performance disparities between test and training data can inadvertently reveal which data points were part of the training set, posing a privacy risk. This study examines how two fundamental aspects of classifier systems—training data quality and classifier training methodology—contribute to privacy vulnerabilities. Our theoretical analysis demonstrates that classifiers exhibit universal vulnerability under conditions of data imbalance and distributional shifts. Empirical findings reinforce our theoretical results, highlighting the significant role of training data quality in classifier susceptibility. Additionally, our study reveals that a classifier’s operational mechanism and architectural design impact its vulnerability. We further investigate mitigation strategies through data obfuscation techniques and analyze their impact on both privacy and classification performance. To aid practitioners, we introduce a privacy-performance trade-off index, providing a structured approach to balancing privacy protection with model effectiveness. The findings offer valuable insights for selecting classifiers and curating training data in diverse real-world applications.

\end{abstract}

\section{Introduction}
\label{intro}

\begin{say}{According to Grand View Research, the global artificial intelligence (AI) market was valued at $93.5$ billion U.S. dollars back in 2021. From 2022 to 2030, the market is expected to grow at a compound annual growth rate (CAGR) of $38.1\%$. This growth can largely be attributed to the “continuous research and innovation directed by the tech giants who are driving the adoption of advanced technologies in industry verticals, such as automotive, healthcare, retail, finance, and manufacturing.”}\end{say}\footnote{https://www.microsourcing.com/learn/blog/the-impact-of-ai-on-business/}

At present, AI-driven classifiers have become indispensable to the functioning of enterprises and businesses \cite{ai-business}. In 2017, $JPMorgan \; Chase \; \& \; Co.$ used an AI tool called Natural Language Processing to review loan applications \cite{jpmorgan}, which saved 3600000 work hours of lawyers. Additionally, it led to a reduction in human errors as well. 
AI models use historical data about operations and render an informed forecast for an upcoming instance. The prediction models find utility in diversified domains like --- sales forecasting \cite{sales-forecast}, marketing response \cite{marketing}, supply chain optimization \cite{supply-chain}, fraud detection \cite{fraud}, customer base expansion \cite{customer1,customer2}, customer segmentation \cite{customer-segment}, customer semantic analysis \cite{sentiment}, churn prediction \cite{churn1}, and financial and investment analysis \cite{risk1,risk2}. It is expected to increase in the future in the context of healthcare support transfer \cite{health-support}, knowledge sharing in virtual communities \cite{virtual-community}, and implementation of smart-home \cite{smart-home} and smart-city paradigms \cite{smart-city}, among many others. 
\\ However, ethical and privacy concerns on the operations of AI-based systems often accompany the upsurge in their deployment \cite{ethics1,ethics2}. A prominent example of ethical concern on AI deployment is Amazon's AI-powered recruitment tool \cite{amazon-ai}. It was deployed to automate the recruitment-related screening. However, it was later found that the tool was discriminating against women as its training data (historical hiring data) was male-dominated. In another incident in 2024, a deepfake automated call impersonating the President of the United States back then was used to influence the election, in which the voters were urged to skip the presidential election. This particular incident involved an AI-based voice cloning and was an attack on individual privacy as well as the democracy of a nation \footnote{https://apnews.com/article/new-hampshire-primary-biden-ai-deepfake-robocall-f3469ceb6dd613079092287994663db5}. Privacy and integrity of data contributors and data owners are two fundamental challenges in this respect. The concerns lie with safeguarding sensitive information, proprietary processes, and customer data from unauthorized access and malicious intentions. The design of web-based platforms inherently supports myriad data collection, often without the consent of the rightful owners. This practice underscores critical issues surrounding privacy, user autonomy, and initiates a lack of ethics in the handling of personal information.

\begin{say}{According to IBM’s Cost of Data Breach Report 2023, the average cost of a data breach reached an all-time high in 2023 of 4.5 million US dollars. This represents a $2.3\%$ increase from the 2022 cost of USD 4.3 million US dollars}\end{say}\footnote{https://www.metacompliance.com/blog/data-breaches/5-damaging-consequences-of-a-data-breach}.
Expenses are incurred from compensating impacted customers, initiating incident response activities, conducting breach investigations, investing in enhanced security measures, covering legal fees, and facing substantial regulatory fines for GDPR non-compliance \cite{expenditure}. In addition to these, such breaches have far-fetched consequences from corporate perspectives, as they involve the \emph{trust} of its customers, which is an intangible resource for any concern \cite{breach-trust1,breach-trust2}. 
To this end, several government and industry standards have been put into place in recent years. The primary goal is to uphold the privacy of individuals \cite{gdpr,ccpa}. We quote from Article 17 of the GDPR, known as the Right to erasure (right to be forgotten).
\\ \emph{The data subject shall have the right to obtain from the controller the erasure of personal data concerning him or her without undue delay, and the controller shall have the obligation to erase personal data without undue delay where one of the following grounds applies:....\footnote{https://gdpr-info.eu/art-17-gdpr/}} 
\\ There can be several ways in which the privacy of individuals gets compromised, out of these, the two main categories are intentional and unintentional breaches \cite{breach}. The former is executed and attempted by malicious entities to achieve some materialistic gain \cite{attack-survey}. The latter is less known and is often caused by the design of the schemes and interfaces, which causes information leakage.
The training phase of a data model is a budding stage of one such breach. During training, information about the training data gets embedded in the model, which creates its footprints in the model, contributing to its vulnerability, which is exploited in the test phase through various types of attacks \cite{yu2014big}. Intruders may attempt to infer the training data \cite{training-data-attack} or the parameters of the model \cite{model-extraction1,model-extraction2}. The former is an attack on individual user data, while the latter is targeted at the revelation of metadata of a model. Alternatively, they can also focus on deceiving the classifier by rendering wrong predictions on query data \cite{adversarial-vuln}. In most cases, the attacks target information acquired during the training phase and are carried out by systematic queries to the model during the test phase. This research investigates the privacy and security concerns in AI deployment from the perspective of \emph{data footprints} left on a classifier.

Prelude to our research:
When a classifier shows a substantial difference in performance between training data and test data, favoring the training data, it can be inferred that the classifier has left its footprints in the model. 
Vulnerability of the classifiers induces a substantial violation of the Right to Erasure act of GDPR because footprints can compromise the integrity of user data, thereby affecting data owners. The attackers can also extract the meta-data of the model through their acquired knowledge of the training data. In this research, we carry out a comprehensive analysis of the vulnerability of different classifiers. The focus of this work is an exploration of the vulnerability of user data both theoretically and empirically. 
The next segment of our research focuses on exploring the degree of vulnerability shown by a large variety of classifiers.\\ 

\noindent \textbf{\emph{RQ: Does a given training data show similar vulnerabilities across different classifiers? Do all types of data exhibit similar vulnerabilities when used to train a classifier?}}
Understanding the vulnerabilities of classifiers can help practitioners make informed decisions. It can enable the business concerns to choose the optimal model that makes a trade-off between privacy preservation and performance. To gain a broader perspective of classifier vulnerability, several aspects need to be delved into. Firstly, we would analyze whether different classifiers show similar vulnerabilities if they show differential vulnerability, and whether there is a relation between vulnerability and overfitting of the classifiers. As training data is an integral component of a classifier, we would explore if an association exists between training data quality and vulnerability of the classifier that it trains. In this research, we conduct different studies to cater to these aspects. We train different classifiers on a given set of data and explore their vulnerabilities. It is also explored if the \emph{classifiers with low overfitting} tend to show lower vulnerabilities. Additionally, we explore datasets of diversified quality on each classifier. Our research indicates the strong presence of a vulnerability in several classifiers, such as the Decision Tree, Random Forest, k-Nearest Neighbor, XGBoost, and Multi-layer perception (artificial neural network). On the contrary, classifiers like Logistic Regression, AdaBoost, Gaussian Naive Bayes, and SVM-like Stochastic Gradient Descent have shown lower vulnerability. Our study also shows that vulnerability present in a classifier is dependent on the underlying training data. When the classes in a dataset possess class imbalance and distributional shift, the vulnerability tendency of the vulnerable classifiers is accentuated. On the contrary, a dataset with equitable class distribution and lacking specific identifiable patterns has been shown to induce lower vulnerability. When deploying trained models in a shared domain, practitioners can leverage the acquired insights to determine the most suitable classifier for end users. This information can serve as a tangible measure to assess the quality and usability of the datasets. It can benefit both buyers and vendors in a data marketplace -- the former can competitively price the datasets based on their vulnerability, while the latter can gain clarity on how well a dataset satisfies the need for privacy \cite{privacy-price}.

The exploration of classifiers and data vulnerability raises some issues that need to be addressed. Our study identifies some \emph{vulnerable} classifiers among the ones explored, and we apply extant ways to reduce the vulnerability of the classifiers. To this end, we explore the usability of data obfuscation - the systematic masking of sensitive attributes to reduce identifiable information without compromising functional values. We investigate the extent to which the data obfuscation schemes remove the footprints of data and mitigate the vulnerability of the {vulnerable} classifiers. We propose to obfuscate the training data before training a classifier. The empirical results show that obfuscation helps reduce the footprints of data, thereby lowering the vulnerability of the classifiers. Now the question arising is, Does masking the training data lessen the efficiency of a model and lead to a fall in performance?  
The last segment of our research is devoted to answering this question.
The goal is to capture the effect of data obfuscation on vulnerability and performance. We define \emph{privacy performance trade-off} to quantify and integrate the change in vulnerability and performance in a single variable following the vulnerability treatment. It measures the admissibility of the changes. The empirical results show that obfuscation causes a fall in the performance of the models, but the degree of change falls within the tolerance level.
\\ In summary, the differential behavior of classifiers on training and test data serves as the foundational context in the study of footprints of data in a classifier. In the main segment of this study, we conceptualize the vulnerability of classifiers and try to find the root causes of vulnerability -- whether it is the architecture of classifiers, the data, or both. A theoretical analysis is presented. Further, an experimental study is conducted, and the outcomes indicate the contribution of both. Further studies analyze the role of data obfuscation in reducing vulnerabilities of classifiers and the resultant fall in predictive performance through the formalization of a definition of privacy-performance trade-off.
\subsection{Organization of the paper}
The rest of this paper is organized as follows. In the next section, we describe the extant works on privacy aspects. In Section 3, we elaborate on the proposed methodology. Sections 4 and 5 elaborate on the experimental components and experimental design. In Sections 6 and 7, we present the results of the experiments and their discussion, respectively. We conclude the paper in Section 7.

\section{Literature Review}
\label{sec:lit}

Privacy is of utmost importance in the digital age; research in this domain traces back to the days of the inception of the internet and the web. Over the decades, the focus on safeguarding personal information and critical web data has amplified with the evolution of technology.  
Privacy is critical for ensuring an ethical environment and sustainability \cite{privacy-sustainability1}. The deluge of ML has called for an increasing need for privacy measures and has prompted nations, bodies, and establishments to formalize privacy requirements on the web \cite{eu-Goodman,eu-privacy}.
Studies manifested the need for efficacious strategies to secure personal data and maintain the integrity of tailored environments \cite{privacy-personalization}. It deals with the protection of sensitive information of the data owners and their privacy rights \cite{privacy-rights1,privacy-rights2}, lowering of data breach risks \cite{data-breach-risk}, and support for ethical AI operations \cite{ethical}. Moreover, as machine learning models increasingly serve as critical organizational assets, safeguarding their metadata has become a vital priority \cite{ml-privacy}. Its security is crucial to maintain competitive advantage and operational integrity.

Initial research in this domain dealt with the theorization of privacy and related concepts \cite{privacy-need,privacy-ethics}.  It involved segregation of different dimensions of privacy, such as unauthorized data collection, use of authorized data in an unintended context, risks of other improper access, and error. The formalization of definitions led to a clear understanding of what is required to address privacy challenges effectively. Early research predominantly concentrated on safeguarding data at rest and in transit, addressing concerns such as unauthorized access, data breaches, and identity disclosure. However, with the advancement of time and the emergence of black-box attacks, \emph{model privacy} also emerged as a necessity \cite{model-privacy}. The shift is driven by an adversary's ability to extract two categories of knowledge -- the metadata of a model, like architecture or model parameters, or the training data used to build the model, or both \cite{blackbox-is}. In both types of attack, the adversary has no direct access to the ML model's design and can only interact with it by submitting test queries. The attacker can gain admissible access to the metadata as well as the training data by inferring from the prediction on carefully crafted test points. With the emergence of cloud computing, as the sharing of models and data became the call for the data, several newer forms of attacks emerged. This new class of attack is a passive attack on the training data as well as the metadata of the model by making naive and permitted queries to the classifiers. This class of attack has a passive modus operandi and occurs without any active breaching effort from the trespasser; hence, it is difficult to detect and mitigate \cite{black-survey}.

On the other side of the same coin, technical researchers have focused on rendering diversified solutions to uphold privacy constraints. Recent advancements have introduced various privacy-preserving techniques for classification, data search, and clustering \cite{privacy-survey,privacy-survey1}. The privacy-preserving solutions depend on the types of threats they aim to mitigate \cite{cyber-threat}. Primitive solutions rendered privacy by anonymizing the sensitive attributes of the data \cite{data-privacy}. This class of methods proved helpful in upholding the integrity of data residing in a shared platform. A fundamental class of safeguards includes cryptographic approaches like homomorphic encryption, elliptic curves, and blockchain \cite{crypto-homo,crypto-elliptic,crypto-blockchain,heidari-blockchain1,heidari-blockchain2}. However, their application to real-world problems is often questionable because of their operations' time and memory requirements. Alternately, the frameworks such as k-anonymity, l-diversity, and t-closeness provided a more feasible way of defining and upholding the privacy of the datasets. In particular, k-anonymity requires the presence of at least k points for a particular attribute, thereby preventing the passive identification of data points \cite{k-anonymity}. The limitations of k-anonymity can be addressed through l-diversity. It requires the presence of at least l instances of a sensitive attribute in each equivalence class of the dataset \cite{l-diversity}. On the other hand, t-closeness calls for a parity between the distribution of a sensitive attribute in the equivalence classes and the entire dataset \cite{t-closeness}. As cloud computation gained prominence \cite{cloud-internet1}, shared computation and shared access of data became a necessity, spurring varied privacy-protection approaches on the platform \cite{heidari-cloud1,darbandi-cloud1,darbandi-cloud2}. The privacy requirements are intensified in a distributed domain, which often deals with sensitive applications like healthcare \cite{distribute-challenges}, energy consumption \cite{energy-consump}. Privacy of data is often ensured through the application of a data obfuscation technique that emphasizes preservation of both computability and privacy of datasets \cite{obfuscation,obfuscation1}. They modify and garble some of the aspects of any dataset by making it illegible to trespassers while preserving its accessibility and usage for valid users. Adopting privacy-enhanced solutions often renders suboptimal computational output in terms of both performance and speed. A deployable solution needs to have a suitable trade-off between the two \cite{privacy-tradeoff}. In recent times, deepfake data generation \cite{heidari-deepfake1,heidari-deepfake2} and generative AI \cite{heidari-chatgpt} caused several privacy breaches, and standard protocols are required to tackle the issue.

Machine Learning-based classifiers are critical assets for any organization \cite{ml-asset}, but they also present dual threats. The exposure of their parameters can lead to direct financial and competitive losses, as adversaries can replicate the model to offer identical services \cite{ml-loss}. Simultaneously, the compromise of the underlying training data can result in equally significant harm, violating user privacy, breaching social and ethical standards, and potentially leading to legal and reputational consequences \cite{data-asset1}. Classifiers of varying types exhibiting a varying degree of information loss.

To this end, we carry out this research to explore this tendency in an array of diversified classifiers.

\begin{figure*}[htbp]
    \hspace*{-1cm}  
    \includegraphics[scale=0.4]{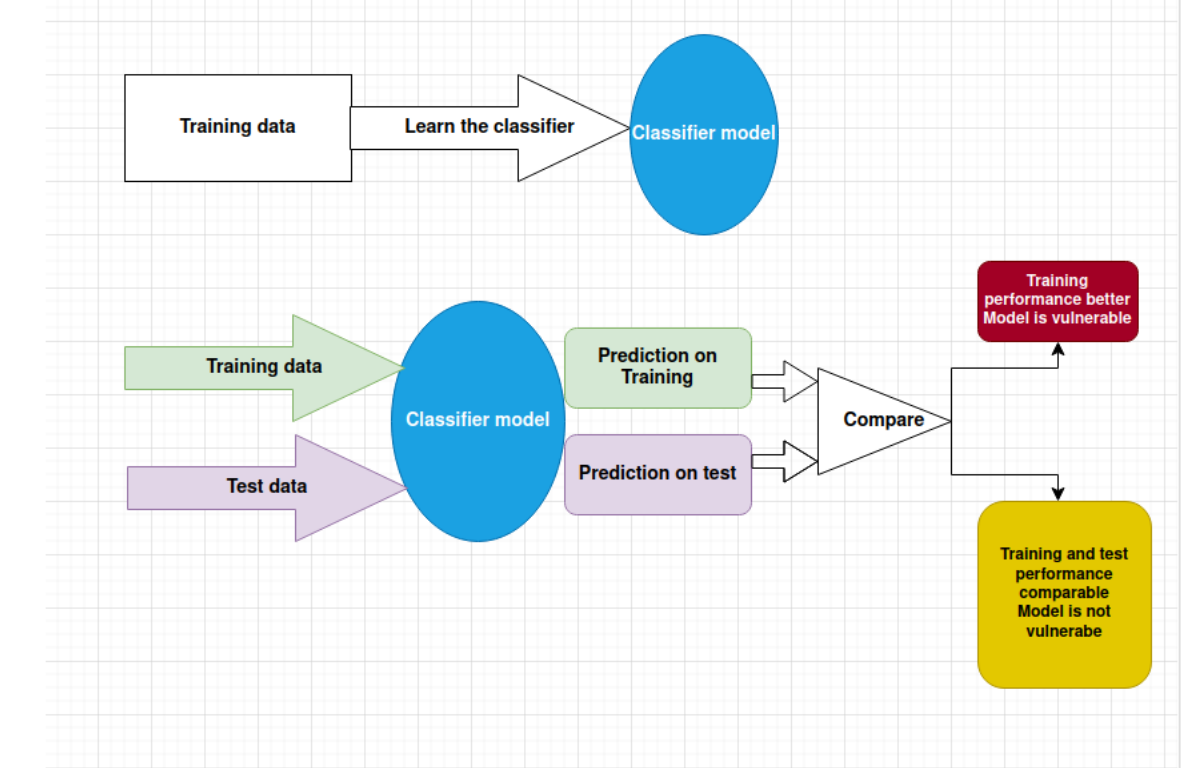}
    \caption{Intuition behind vulnerability estimation. Training data is used in building a classifier. When the predictions obtained for the training data and the test data show similar veracity, the classifier is said to be not vulnerable. If the veracity of prediction is substantially higher for the training data, the classifier is identified as vulnerable.}
\end{figure*}

\section{Methodology}
\label{sec:meth}
In recent years, business ventures employed classifiers to automate substantial segments of internal and external operations \footnote{https://www.ust.com/en/insights/revolutionizing-document-classification-with-trainable-ai-models} \cite{ai-changing}. These models are developed or \emph{trained} on one or several datasets, which are usually known as \emph{training data}. The data used to build classifiers is a valuable asset \cite{data-asset} — not only for the businesses that rely on it, but also for the individuals whose personal information is contained in it. Both personal data and the classifier's meta are critical, and their protection is essential, as it impacts both privacy and competitive advantage \cite{privacy-data-asset} The emphasis of this work is privacy of \emph{data used in training (personal data)}, which can get compromised through the vulnerability of the classifier which it trains.
\subsection{Recognizing the vulnerabilities}
\begin{itemize}
\item \textbf{Vulnerability of a classifier}: 
Classifiers are multifaceted mathematical frameworks designed to generalize patterns from data while minimizing a defined error function \cite{classifier-statistics}. The error function quantifies the disagreement between the model's predictions and the true labels in the training set, serving as a guide for optimizing the model. However, the model itself encompasses aspects of architecture (and algorithm) and their related parameters and hyperparameters, and also the training data that is used to develop the model \cite{hastie-book}. During training, the algorithms iteratively adjust the model's parameters to minimize this error function on the training data. However, the goal is not just to memorize training data. Instead, good classifiers are supposed to balance minimization of training data error with generalization to unseen data, which is formalized through techniques like regularization, cross-validation, and bias-variance trade-offs \cite{islr}.

A classifier $f_\theta: \mathcal{X} \to \mathcal{Y}$ with parameters $\theta \in \Theta$ is trained by minimizing an error function $L(\theta)$ over a training dataset $\mathcal{D}_{\text{train}} = \{(x_i, y_i)\}_{i=1}^N$:

\begin{equation}
\theta^* = \argmin_{\theta \in \Theta} L(\theta) = \argmin_{\theta \in \Theta} \frac{1}{n} \sum_{i=1}^n \ell(f_\theta(x_i), y_i) + \lambda R(\theta)
\end{equation}

where,
\begin{itemize}
    \item $\ell: \mathcal{Y} \times \mathcal{Y} \to \mathbb{R}^+$ measures the loss incurred per data point.
    \item $R(\theta)$ is a regularization term which controls the trade-off between embedding on training data and capability to generalize on test data.
    \item $\lambda \geq 0$ controls the strength of regularization.
\end{itemize}

Here, $\theta^{*}$ defines the classifier. It is derived from the training data ($\mathcal{D}_{\text{train}}$) and the configuration of the model by repeated iterations of learning. As $\theta^{*}$ is dependent on the training data, there is a possibility that the model fortuitously includes details of the training data in the model, which the adversaries can exploit \cite{shokri-privacy}.

In the business domain, parties with conflicting interests can attempt to learn the training data of their adversaries. The parties can involve any curious individual attempting to learn sensitive parameters of a particular user(s) or a competing business venture. The interest of a competing business venture can be many-fold; it can attempt to destabilize the operation of its rival venture, or it may attempt to gain profit by learning the data of its rival. 
\\ A trespasser can attempt to guess and learn the training data without making an explicit query about it. This is possible by virtue of \emph{inference attacks} on the prediction model \cite{shokri-membership}. By obtaining the prediction of a particular data point, a trespasser can know whether it served as training data. In this scenario, it is essential to quantify the vulnerability of a set of training data points. Usually, the accuracy and correctness of the training data points are higher than that of the unseen (test) points. The trespassers use this information to make inferences about point memberships in the training set. We quantify and define the vulnerability of the training data on the basis of these intuitions. The definition is as follows.
    \begin{dfn}(Vulnerability of classifier $vul(\mathcal{D})$): It is quantified by the ratio of the prediction performance of the training data points to that of the data points for a dataset ($\mathcal{D}$).
        \end{dfn}
      \begin{equation}
          vul(\mathcal{D})=\frac{\textrm{Correctness of predictions on training data for $\mathcal{D}$}}{\textrm{Correctness of predictions on test data for $\mathcal{D}$}}
      \end{equation}
The more the difference in the correctness of prediction for the training set and the test set, the severe the vulnerability of the classifier. A higher vulnerability provides a favorable substrate for trespassers to guess the training data correctly and permits leakage of data. When the veracity of predictions is similar for training and test data, it is difficult for a trespasser to guess the training data points correctly. Figure 1 presents a visual representation of this idea. It shows that training data is used to model a classifier. Predictions are obtained for the training data as well as the test data from the classifier. If there is a substantial difference in performance between the training and the test data (with the training data being better), the classifier is said to possess vulnerability.
\\ We have to be judicious with the choice of assessment metric in different cases. In the generic scenario, the accuracy of classification provides admissible information. However, accuracy can be insufficient to deduce a picture when the dataset possesses irregularities like imbalance and missing features (there can be more). For example, accuracy can give misleading results on imbalanced data. The true picture can be given through metrics like minority class $F_1$, precision, and recall \cite{imbalance-precision}.

The two key aspects of a classifier are -- i] its training data, and, ii] its structure and configuration of parameters (for example, a Decision Tree classifier has a splitting criteria of decision making at each node, whereas a Neural Network-based classifier has a set of hidden nodes and weights to define the decision boundary between the classes). To probe the aspect of classifier vulnerability further, we investigate these two aspects. We seek answers to whether the underlying training data contributes to a classifier's vulnerability and/ or the framework of the classifier.
\item \textbf{Is a given training data's vulnerability the same across different models?}
\\ Let us consider two different classifiers, A and B. If the difference in training and test performances is low for classifier A than B, it is favorable to deploy A. Even though A and B are trained on the same training data, A will divulge less information about the training data than B, as A remembers fewer details about the training data. This exploration can provide ways to secure the training data with the available options (classifier options), without the involvement of any external add-on.

We carry out a detailed experimental study to find answers to the question. For each dataset in the study, an array of classifiers is tested for vulnerabilities on various datasets. The findings will elucidate the classifier framework's influence on the vulnerability of classifiers. 
\item \textbf{Is there a correlation between a classifier's interpretability and the training data's vulnerability?}
\\ Interpretability of predictive models is a noteworthy aspect nowadays. It refers to the ease with which humans can understand and also explain the decisions rendered by a machine learning or a statistical model \cite{interpretable-isf}. Privacy and vulnerability are also intertwined with human perception of a model. Against this backdrop, the correlation (possible) between classifier vulnerability and classifier interpretability is worth exploring. We carry out an empirical study to explore if there is a positive, negative, or no correlation between the two.
\item \textbf{Does a classifier exhibit similar vulnerabilities when trained on distinct datasets with varying distributions?}
\\ Since training data is a fundamental component of a classifier, we investigate the influence of its quality and distribution on a classifier's vulnerability. We investigate if the leakage of training data from a classifier (through test point queries) is dependent on itself. In this context, we present a theoretical analysis of the universal vulnerability of classifiers under class imbalance and distribution shift below.

\textbf{Universal vulnerability of classifiers under class imbalance and distribution shift} 

Despite their use in supervised learning problems, classifiers are inherently vulnerable to class imbalance and changes in data distribution between training and testing data. This is mainly applicable for classifiers trained under the empirical risk minimization (ERM) principle, which aims to minimize the average loss over the training dataset, assuming it reflects the true data distribution \cite{chakraborty2020superensemble}. Classifiers that follow the ERM principle include logistic regression, support vector machines, decision trees, neural networks, k-nearest neighbors, naïve Bayes, and ensemble methods like random forests and gradient boosting \cite{shalev2014understanding}\cite{chakraborty2019nonparametric}.\\

\noindent Given a training set $\left\{\left(x_i, y_i\right)\right\}_{i=1}^n \sim P_{\text {train }}(X, Y)$, and a loss function $\ell(f(x), y)$, the ERM principle seeks a classifier $f \in \mathcal{F}$ (from a hypothesis class $\mathcal{F}$) that minimizes:
\begin{equation}
f^*=\arg \min _{f \in \mathcal{F}} \hat{R}_n(f), \quad \text { and } \quad \hat{R}_n(f)=\frac{1}{n} \sum_{i=1}^n \ell\left(f\left(x_i\right), y_i\right),
\end{equation}
where $\hat{R}_n(f)$ is the empirical risk (e.g., average training loss). Any classifier $f: \mathcal{X} \rightarrow \mathcal{Y}$ is trained using a training data distribution $P_{\text {train }}(X, Y)$ is expected to generalize under a potentially different test distribution $P_{\text {test }}(X, Y)$. The expected risk on the test distribution is given by:
\begin{equation}
R_{\text {test }}(f)=\mathbb{E}_{(X, Y) \sim P_{\text {test }}}[\ell(f(X), Y)],
\end{equation}
where $\ell$
is a bounded loss function (e.g., 0-1 loss). We now theoretically show that no classifier can guarantee consistent generalization performance unless it is robust to shifts in class priors, covariates, or label distributions, following \cite{ben2010theory}. Therefore, addressing imbalance and distributional shift is not merely a practical concern but a fundamental theoretical necessity in supervised learning. In Proposition \ref{prop}, we theoretically show that the expected test risk $R_{\text {test }}(f)$ can differ substantially from the training risk $R_{\text {train }}(f)$ under class imbalance, covariate shift, or concept shift, thus rendering any classifier theoretically vulnerable. Before stating the theoretical result, we mathematically discuss the conditions under which class imbalance, covariate shift, and concept shift cause universal vulnerability for any pattern classifiers. 

\begin{dfn}(Class Imbalance/Prior Shift)
     The class prior in training is skewed, i.e., $P_{\text {train }}(Y=1) \ll P_{\text {train }}(Y=0)$, and $P_{\text {test }}(Y)$ may differ from $P_{\text {train }}(Y)$, even if $P(X \mid Y)$ remains unchanged.
\end{dfn}

\begin{dfn}(Covariate Shift)
     The marginal distribution over features changes: $P_{\text {train }}(X) \neq P_{\text {test }}(X)$, while $P(Y \mid X)$ remains invariant. 
\end{dfn}

\begin{dfn}(Concept Shift)
     The conditional label distribution changes across domains:
$ P_{\text {train }}(Y \mid X) \neq P_{\text {test }}(Y \mid X).$
\end{dfn}
\noindent Let $\mathcal{X}$ be the input space and $\mathcal{Y}=\{0,1\}$ be the label space.Under the following types of shifts, we have the main theoretical result:
\begin{pro}\label{prop}
Let $f: \mathcal{X} \rightarrow \mathcal{Y}$ be a classifier trained via empirical risk minimization on $P_{\text {train }}(X, Y)$, and evaluated on $P_{\text {test }}(X, Y)$. Then, under any of the conditions defined above (Definitions 2-4), we have:
$$\left|R_{\text {test }}(f)-R_{\text {train }}(f)\right|>0$$
indicating degraded generalization performance due to the mismatch between training and test distributions. That is, classifiers trained using empirical risk minimization on the training distribution are not universally robust to class imbalance or distribution shift.
\end{pro}
\begin{proof}{.}
We aim to show that under any of the three distributional shift conditions, the expected loss of the classifier differs between training and test domains.\\

\noindent Let $\ell: \mathcal{Y} \times \mathcal{Y} \rightarrow \mathbb{R}_{+}$be a pointwise loss function (e.g., 0-1 loss). Define the expected risk under any distribution $P$ as:
\begin{equation}
R_P(f)=\mathbb{E}_{(X, Y) \sim P}[\ell(f(X), Y)]=\int_{\mathcal{X} \times \mathcal{Y}} \ell(f(x), y) d P(x, y).
\end{equation}
We compare $R_{\text {test }}(f)$ and $R_{\text {train }}(f)$ by examining how the joint distributions $P_{\text {train }}$ and $P_{\text {test }}$ differ under each of the three cases.\\

\noindent Case 1 (Prior Shift / Class Imbalance):
Assume $P_{\text {train }}(Y) \neq P_{\text {test }}(Y)$ but $P(X \mid Y)$ remains invariant. Then, we have
\begin{equation}
P_{\text {train }}(x, y)=P_{\text {train }}(y) P(x \mid y) \quad \text{and} \quad P_{\text {test }}(x, y)=P_{\text {test }}(y) P(x \mid y).
\end{equation}
Hence, the difference in risk becomes:
\begin{equation}
R_{\text {test }}(f)-R_{\text {train }}(f)=\sum_{y \in \mathcal{Y}}\left[P_{\text {test }}(y)-P_{\text {train }}(y)\right] \int_{\mathcal{X}} \ell(f(x), y) P(x \mid y) d x.
\end{equation}
This is generally non-zero unless the loss term $\ell(f(x), y)$ is constant across all $y$, which is highly unlikely in practice, especially under class imbalance, where majority/minority label errors are systematically different.\\

\noindent Case 2 (Covariate Shift): Assume $P_{\text {train }}(X) \neq P_{\text {test }}(X)$, but $P(Y \mid X)$ remains unchanged. Then, we have
\begin{equation}
P_{\text {train }}(x, y)=P_{\text {train }}(x) P(y \mid x) \quad \text{and} \quad P_{\text {test }}(x, y)=P_{\text {test }}(x) P(y \mid x).
\end{equation}
Hence, the difference in risk becomes:
\begin{equation}
R_{\text {test }}(f)-R_{\text {train }}(f)=\int_{\mathcal{X}}\left[P_{\text {test }}(x)-P_{\text {train }}(x)\right] \sum_{y \in \mathcal{Y}} \ell(f(x), y) P(y \mid x) d x
\end{equation}
Unless the weighted loss $\sum_y \ell(f(x), y) P(y \mid x)$ is constant over $\mathcal{X}$, the above expression is non-zero for any mismatch between training and test covariate distributions.\\

\noindent Case 3 (Concept Shift):
Here $P(Y \mid X)$ itself changes, thus, we have
\begin{equation}
P_{\text {train }}(x, y)=P_{\text {train }}(x) P_{\text {train }}(y \mid x) \quad \text{and} \quad P_{\text {test }}(x, y)=P_{\text {test }}(x) P_{\text {test }}(y \mid x).
\end{equation}
Even if $P_{\text {train }}(X)=P_{\text {test }}(X)$, the difference in conditional distributions yields:
\begin{equation}
R_{\text {test }}(f)-R_{\text {train }}(f)=\int_{\mathcal{X}} P(x) \sum_{y \in \mathcal{Y}} \ell(f(x), y)\left[P_{\text {test }}(y \mid x)-P_{\text {train }}(y \mid x)\right] d x.
\end{equation}

\noindent This is generally non-zero when the classifier's predictions disagree with the new test conditional $P_{\text {test }}(y \mid x)$.\\

\noindent In all three cases, the test risk $R_{\text {test }}(f)$ differs from the training risk $R_{\text {train }}(f)$, implying that empirical risk minimization fails to guarantee robustness under common forms of distribution shift and imbalance.
\end{proof}

\noindent In practice, we can assume that
$
R_{\text {test }}(f) \geq R_{\text {train }}(f)
$
due to overfitting and distribution mismatch, especially in the presence of class imbalance or shift. 

\begin{remark}
    This proposition underscores a critical limitation of classifiers trained under the empirical risk minimization paradigm. In real-world settings, such distributional mismatches are common, as seen in Sec. \ref{exp_sec}. As a result, classifiers may exhibit severely degraded performance when deployed, especially on underrepresented groups, edge cases, or in dynamically changing environments \cite{shalev2014understanding}. This highlights the necessity for robust learning strategies such as reweighting, resampling, domain adaptation, or distributionally robust optimization to mitigate these vulnerabilities and ensure safe and fair deployment of machine learning models, which has been explored in \cite{chakraborty2020hellinger}\cite{cao2019learning}\cite{duchi2023distributionally}.
\end{remark}


\end{itemize}
    
\subsection{Extant ways to uphold privacy}  
Vulnerability of classifiers is a prevalent phenomenon that can be exploited by adversaries to gather information about training data. Given that, we must devise intuitive ways to reduce the training data's vulnerability. We explore data obfuscation techniques for the same.
\begin{itemize}
    \item \textbf{Data obfuscation for mitigation of classifiers' vulnerability}
    Data obfuscation techniques consist of a suite of privacy-preserving techniques to garble the sensitive attributes of a data without compromising its utility \cite{obfuscation}. They are particularly useful in scenarios where we need to share details but also run a potential risk of revealing sensitive aspects. Data obfuscation can be achieved in several ways --- perturbation, anonymization, masking, and randomization. Data obfuscation renders real-looking data to fool trespassers interested in learning the data. Data before and after obfuscation look similar; hence, it is often difficult for someone to distinguish between the two.

    In this research, we explore whether obfuscating the training data helps us hide the training data's signatures in the classifier. The goal is to find out the extent to which data obfuscation diminishes a classifier's capability to distinguish between the training data and test data. We would explore how the training data vulnerability, $vul(\mathcal{D})$ (defined in the previous section), changes with the incorporation of data obfuscation. We define ${vul(\mathcal{D})}_{change}$ in this regard. Let $vul(\mathcal{D})$ and ${vul(\mathcal{D})}^{obf}$ be the vulnerability obtained for a dataset using its regular and obfuscated version, respectively.
    \begin{dfn}(Change of vulnerability, ${vul(\mathcal{D})}_{change}$): It is quantified by the ratio of percentage change of $vul$ after obfuscation with respect to the $vul$ obtained on the regular (not obfuscated) data.
    \end{dfn}
      \begin{equation}
      {vul}_{change}=\frac{vul(\mathcal{D})-{{vul}_{obf}}(\mathcal{D})}{vul(\mathcal{D})}
      \end{equation}
       
      The range of ${vul}_{change}$ can vary between $-\infty$ to $\infty$. When ${vul}_{change}$ has a negative value, it signifies that data obfuscation worsened the training data's vulnerability. A positive value of ${vul}_{change}$ signifies a constructive development, as it indicates the lessening of the vulnerability. The zero value manifests the obfuscation technique's failure to render an improvement. Note that while computing, ${{vul}_{change}}(\mathcal{D})$, we consider the same classifier for consistency and training points for $vul(\mathcal{D})$ and ${{vul}_{obf}}(\mathcal{D})$ computation.
   
\end{itemize}
    
\subsection{Exploring the privacy-performance trade-off}
The privacy-performance trade-off forms a critical bottleneck in the design of deployable models on shared platforms \cite{privacy-tradeoff}. While it is necessary to share data that can render meaningful computational outcomes to all parties, protecting the data is also obligatory from the data owners' perspective \cite{privacy-necessity}. While ensuring data privacy, practitioners and model developers are often bothered about the loss in performance quality. We investigate the following contexts.
\begin{itemize}
    \item \textbf{Association between vulnerability and performance}: It is necessary to explore if some association exists between the degree of vulnerabilities and performance on a classifier. We need to explore further if two datasets with different $vul$ (on the same classifiers) render similar or differing performances; if yes, is there any association or not?
    \item \textbf{Comparing the performance of the test data with and without data obfuscation}:
    In this part, we propose to use data obfuscation techniques to reduce the vulnerability of the training data. Apart from the privacy aspect, we need to investigate whether data obfuscation degrades the predictive capability of a classifier. In the following definition, we try to capture the privacy-performance
    trade-off. We denote the performances of data $\mathcal{D}$ before and after obfuscation with $perf(\mathcal{D})$ and ${perf}_{obf}(\mathcal{D})$ and propose the following definition for quantifying the privacy-performance trade-off.
     \begin{dfn}(privacy performance trade-off ${PP}_{trade-off}(\mathcal{D})$): It is quantified by the product of the ratio of change in the data vulnerability before (${vul}(\mathcal{D})$) and after obfuscation (${vul}_{obf}(\mathcal{D})$) and the ratio of change in performance after (${perf}_{obf}(\mathcal{D})$) and before data obfuscation (${perf}(\mathcal{D})$).
     \end{dfn}   
      \begin{equation}
      {PP}_{trade-off}(\mathcal{D})=\left(\frac{vul(\mathcal{D})}{{vul}_{obf}(\mathcal{D})}\right) \times \left(\frac{{perf}_{obf}(\mathcal{D})}{{perf}(\mathcal{D})}\right) 
      \end{equation}
        The first component quantifies the change of data vulnerability due to the incorporation of data obfuscation; the more significant the decrease in vulnerability, the stronger the privacy-preservation. The second component quantifies the change in performance; the lesser the fall in performance, the better it is. Hence, we have taken the ratio in reverse order for the two. Consequently, the higher the value of ${PP}_{trade-off}(\mathcal{D})$, the better the outcome.
        \\ Various data obfuscation protocols are found in the literature \cite{obfuscation-review}. For a given data $\mathcal{D}$ and a classifier, the protocol can be chosen on the basis of ${PP}_{trade-off}(\mathcal{D})$ values. The practitioners can choose the obfuscation protocol that gives the highest ${PP}_{trade-off}(\mathcal{D})$ for a given data.
    
    We explore the association between (i) reduction of data vulnerability with data obfuscation and (ii) performance trend with data obfuscation.  
\end{itemize}
\section{Experimental components}\label{exp_sec}
Several experiments are carried out in this section to explore the vulnerabilities of the datasets and the classifiers. The methodologies discussed in the previous section are evaluated empirically and discussed here.

\subsection{Datasets}
The empirical study is carried out on five datasets: four (Body, Contraceptive, Customer, and Optdigits) are multiclass, and the remaining one (Churn) is binary and class imbalanced. There is a slight degree of imbalance among the classes of the Contraceptive and Customer datasets. On the other hand, the classes Body and Optdigits are well-balanced. The datasets are described below.

\begin{itemize}
    \item \textbf{Churn prediction (insurance) \footnote{https://www.kaggle.com/datasets/k123vinod/insurance-churn-prediction-weekend-hackathon}}: The dataset information about 33908 customer profiles with 16 features capturing insurance utilization patterns, serving as a binary classification dataset for churn prediction (policy cancellation/ non-renewal). There is an imbalance ratio of 8.55 between the churners and non-churners, where the latter outnumber the former. The feature set likely contains both continuous variables and categorical features, typical of a real-world dataset. The dataset source provides a training set and a test set, where the class memberships are known for the training partition, but only for the test partition. Hence, we have used the training partition for our analysis. Investigation of such datasets is important from financial and marketing perspectives as it can provide insights into actuarial risk modeling and formulating marketing strategies.
    \item \textbf{Contraceptive survey \cite{uci-contra}}:
    This dataset is a subset of the 1987 National Indonesia Contraceptive Prevalence Survey. This data collection is aimed at understanding and correlating the socio-economic and demographic factors with the choice of contraceptive method. The choice of contraceptives was --- i] no choice, ii] long-term method, and iii] short-term method. The data was collected from 1473 individuals, each corresponding to a non-pregnant or not known to be pregnant woman in Indonesia. There are nine features which provide socio-economic and demographic information, and we have both numeric and categorical features. The dataset is multi-class, and we have to classify each point (woman) into any one of these three classes -- no choice (class 1), long-term method (class 2), and short-term method (class 3) of their choice of contraceptive. The three classes possess a modestly unbalanced distribution with 629, 333, and 511 points in classes 1, 2, and 3, respectively.
    \item \textbf{Profitable customer segment prediction \footnote{https://www.kaggle.com/datasets/tsiaras/predicting-profitable-customer-segments}}: In marketing, finding the right target customer group for maximizing returns is a critical task. This dataset originates from an online retail company, which has compiled historical data on various customer segments. It tracks the profitability of each group following specific marketing campaigns and evaluates the effectiveness of investing in marketing for these groups in hindsight. There are 70 numeric features extracted from their market research findings for 6620 points, where each feature corresponds to only one group or a pair of groups. Each data point has information pertaining to the comparison of two different groups. It is a three-class dataset where the class of the datapoints predicts whether both groups are profitable, only one is profitable, or neither. The number of data points for three classes is 1893, 1321, and 1353, respectively. 
    \item \textbf{Body performance prediction\footnote{https://www.kaggle.com/datasets/kukuroo3/body-performance-data/data}}: This dataset was collected by the Korea Sports Promotion Foundation to explore the correlation between the body statistics of individuals (features) and their performances (class). The former consists of information like age, gender, height, weight, systolic and diastolic blood pressure, grip force, sit and bend forward distance, sit-up count, and broad-jump distance of an individual. Along with these, the performance of the individual is noted, and it can be classified into any of the four classes. There are 11 features and 13393 data points in this dataset. Each of the four classes has approximately equal shares of $25\%$ points (balanced dataset). This dataset contains the sensitive health parameters of many individuals.
    \item \textbf{Optical digits\footnote{https://archive.ics.uci.edu/dataset/80/optical+recognition+of+handwritten+digits}  \cite{optical}}: This dataset consists of features derived from handwritten digits, each of which belongs to exactly one of ten classes (where each class corresponds to a digit). There are 5620 points, each corresponding to the image of a handwritten digit. There are 64 features, and they are derived by breaking the image into segments and extracting their pixel-related information. This dataset is chosen for its highly balanced nature of representation of the classes.
\end{itemize}
We have considered several factors while choosing the datasets for this study -- i] binary/ multiclass nature, ii] the degree of balanced representation of points, iii] variability in feature cardinality, and iv] the ratio of points to features in a dataset. We have chosen four multiclass datasets and a binary dataset, which is considerably imbalanced (with an imbalance ratio of 8.55). Additionally, two of the four multiclass datasets show a high level of balance in the class distribution. The number of features varies between 9 and 70, and the ratio of points to features varies from 94 to 1028. These variabilities ensure a comprehensive assessment while mitigating any selection bias. While performing the experiments, a dataset is partitioned in the ratio 4:1; four parts are used for training, and the remaining part is used for testing. To explore the vulnerability and obtain the prediction on seen data, we consider a randomly selected $50\%$ fraction of the training set for prediction. 
\subsection{Classifiers}
All the experiments in this work involve assessing the predictive performance of the classifiers. To make an unbiased study, we considered several factors that influence the modus operandi of a classifier. While choosing the classifiers, we have taken an equitable share of white box classifiers, semi-white box classifiers, and black box classifiers \cite{white-black}. It is done to mitigate the interpretability bias of the classifiers, as there is a common perception that interpretable classifiers render suboptimal performance. We have explored nine classifier types \cite{classifiers} -- two with high interpretability (Decision Tree and Gaussian Naive Bayes), three with moderate interpretability (k-Nearest Neighbor, Logistic Regression, and Stochastic Gradient Descent), and four with low interpretability (Random Forest, XGBoost, AdaBoost, and MLP). These classifiers can also be segregated as --- tree-based (Decision Tree, Random Forest, AdaBoost, and XGB), non-linear (Gaussian Naive Bayes, k-Nearest Neighbor, and MLP), and linear (Logistic Regression and Stochastic Gradient Descent). We may note that tree-based classifiers render a non-linear decision boundary. Our analysis includes ten classifiers, comprising two versions of MLP (shallow and deep) along with eight other models. For all the classifiers, we have used their sklearn implementation in Python \cite{sklearn-classifiers}. The working principles of the classifiers and the parameters are given below.
\begin{itemize}
\item \textbf{Decision Tree \cite{dec-tree,scikit}}: Its modus operandi is based on recursive partitioning of the feature space into disjoint regions, with each division made based on the value of a particular feature. The process continues till optimized results are obtained or the stopping criteria are met. Gini impurity is used to split the nodes. We have set the maximum depth to 10 and the maximum feature size to 5.
\item \textbf{Random Forest \cite{random-org,scikit}}: It is an ensemble learning method where the base classifier is a decision tree. It combines the predictions of multiple individual trees to increase the robustness of the classifier system. We used Gini impurity to split the tree, and bootstrap samples were used. Like the previous classifier, the maximum depth and maximum feature size are set to 10 and 5, respectively. The number of estimators (trees) is set to 100.  
\item \textbf{XGBoost \cite{xgboost-org,scikit}}: It consists of an ensemble of Gradient Boosting Machines.  The learning framework is dependent on improving and optimizing past learning mistakes. It is designed to be immune to the overfitting problems of the decision trees. We have set the maximum depth as six and the minimum child weight as one, and all other parameters are set to default values and choices. 
\item \textbf{k-Nearest Neighbor \cite{knn-org,scikit}}: It classifies a test data point by taking a majority vote of its $k$ closest points, assigning the most common class among its $k$ nearest neighbors to the test data point. The neighborhood size $k$ is set to 5.
\item \textbf{Stochastic Gradient Descent \cite{stochastic-org,scikit}}: It consists of regularized linear models ( like SVM or logistic regression) with SGD training. It supports different types of loss functions and incremental learning, such as partial fit, to learn the data. We have used the Python implementation at its default configuration (hinge loss, l2 penalty, regularization strength = 0.0001, and maximum number of iterations = 1000).
\item \textbf{AdaBoost \cite{adaboost-org,scikit}}: It is an ensemble of classifiers. The original dataset fits the first classifier and the weights of instances for subsequent classifiers are adjusted according to the misclassification in the previous classifiers. We have used the Python implementation at its default configuration.
\item \textbf{Gaussian Naive Bayes \cite{gnb,scikit}}: It is a variant of the Naive Bayes Classifier with the additional assumption of Gaussian distribution of the features. It works well for a large number of points and features. We have used the Python implementation at its default configuration.
\item \textbf{Logistic Regression \cite{log-reg,scikit}}: It learns a probabilistic output of class membership by determining the relationship between the features and the predictor variables. The probabilities can be converted to class labels using a variety of methods. We have used the Python implementation with Liblinear solver.
\item \textbf{Multi Layer Perceptron \cite{mlp,scikit}}: It is a powerful but inherently opaque model. Its multi-layer architecture enables it to learn non-linear boundaries, but the latent parameters provide an inherent murkiness in its modus operandi. We have varied the number of layers and number of neurons to generate shallow and deep models, then studied their relation with model vulnerability. This experimental framework allows us to quantitatively assess the influence of architectural choices on a model’s susceptibility to adversarial attacks and data leakage. There are two versions: \emph{MLP-shallow} -- with a single hidden layer of 16 neurons, and \emph{MLP-deep} -- with 8 hidden layers consisting of 128, 64, 32, 32, 16, 16, 8, 8 neurons respectively. We have curated additional experiments to study the effect of the number of layers and the number of neurons on vulnerability.
\end{itemize}

\subsection{Evaluating metrics}
We have used \emph{accuracy} \cite{accuracy} to evaluate the performance of multiclass datasets, and \emph{average precision score} {\cite{avg-precision} for the binary, imbalanced dataset. We have preferred \emph{average precision score} over $F$ scores as the former captures the ranking of the positive instances over negative ones (even when the classifier fails to capture the rare instances) while the latter relies on yes/ no prediction.

\section{Design of experiments}

\subsection{Experiment 1:}
This is the main experiment of this research. In this experiment, we explore the vulnerability of the training data considering different datasets across various classifiers. There are three major goals of this experiment, which are described as follows:
\begin{enumerate}
    \item Does training data vulnerability exist? Are the classifiers differentially efficacious across the data they had seen during training and the data they had not? Vulnerability ($vul$) is calculated from Equation 2.
    \item Is a particular dataset equally vulnerable across different classifiers? Knowledge on this aspect can help practitioners select the privacy-preserving classifiers before deploying the models. For example, let the task be to choose a model (for deployment) that is trained on the \emph{Churn} dataset. A judicious choice for the practitioner would be to choose the classifier with the lowest $vul$ score. 
    \\ Further, can we segregate a classifier as \emph{vulnerable} or \emph{not vulnerable}?
    \item Is a classifier equally vulnerable on different datasets? Does it manifest similar $vul$ scores for different datasets? This analysis can help practitioners choose datasets that will train a deployable model, particularly when the choice of classifier is fixed. For example, if a kNN classifier has to be deployed, the practitioner has to select the dataset that renders the least vulnerability on the kNN classifier and train the required model.
\end{enumerate}

\subsection{Experiment 2:}
This experiment is designed to investigate the role of data obfuscation in curtailing classifier vulnerability. We explore both the \emph{vulnerable} and the \emph{not vulnerable} classifiers. We compare the vulnerabilities of (dataset, classifiers) pairs before and after applying data obfuscation and report the percentage change (Equation 12). \\ Two data obfuscation techniques are explored — (i) LSH-based encoding \cite{lsh} and (ii) Hamming encoding \cite{hamming}. We report their capability to reduce data vulnerability. LSH encoding is based on generating random hyperplanes in a given set of features. The points are projected onto the hyperplanes and given a 0/1 encoding based on their position with respect to each hyperplane. Hamming encoding employs a threshold-based binarization strategy. To encode a point with respect to a feature, the feature value is compared with the average; if it matches or exceeds the average, the feature value is encoded as 1, if not, it is encoded as 0. The process is repeated for all features. We have chosen these two protocols for data obfuscation as they do not require any additional information to generate a classifier \cite{shamsi-privacy}. Although the generated classifier is an approximate one (with respect to the one generated on the original dataset), its obfuscated training data remains untractable from an unauthorized user.

\subsection{Experiment 3:}
After exploring the role of data obfuscation in reducing the vulnerability, we explore a critical aspect related to its usability in practical scenarios. Data obfuscation can find relevance in predictive scenarios if the fall in performance is within an admissible range. To this end, we explore the classifier performances in the original data and the obfuscated data, and integrate them with the change of vulnerability. \emph{privacy performance trade-off} is used for the same (Equation 13). Privacy preservation on a dataset is concedable only when the application of obfuscation renders a high privacy-performance trade-off score.

\section{Analysis and Results}
In this section, we present the outcomes of the experiments on footprint detection. The results can equip the practitioners with insights into the sophisticated aspects of privacy preservation while designing classification models, choosing righteous datasets to train a model, and render a balanced choice in terms of performance as well as privacy.

\subsection{Experiment 1:} Figure 2 shows the fundamental result of this study: whether training data leaves footprints on the classifier. Each column of Figure 2 corresponds to a column and the data shows that training data leaves footprints on some particular classifiers, namely, Decision Trees, Random Forests, k-nearest neighbor, and MLP-deep. On these classifiers, for some datasets, the performances on training and test data have shown significant variation (with training data being better), which can be exploited by a trespasser to successfully infer the training data. For instance, if a classifier consistently predicts the true class for a cluster of points within a small neighborhood, it strongly suggests the presence of a training point in that region
The remaining five classifiers—Stochastic Gradient Descent, AdaBoost, Gaussian Naive Bayes, Logistic Regression, and MLP-shallow—have shown similar performances (accuracy or average precision) on training and test data across all datasets. In this scenario, it is not easy to single out the footprints of training data from the correctness of the predictions. 
On the basis of the empirical finding, we segregate the classifiers as \emph{vulnerable} (Decision Trees, Random Forests, k-nearest neighbor, and MLP-deep) and \emph{not vulnerable} (Stochastic Gradient Descent, AdaBoost, Gaussian Naive Bayes, Logistic Regression, and MLP-shallow).

Next, we explore the row-specific data to study the influence of datasets' characteristics on vulnerability. We have considered five datasets -- \emph{Churn}, \emph{Contraceptive}, \emph{Customer}, \emph{Body} and \emph{Optdigits}. Each row of Figure 2 corresponds to a dataset, and the analysis shows that the degree of vulnerability varies across datasets, as well. Three datasets -- Churn, Contaceptive, and Customer possess a distributional shift in their classes and have exhibited a notable amount of vulnerability across some of the classifiers. The moderate to severe imbalance in their class representation leads to a distribuional bias in the data itself. The bias leads to an overlearning of the training data, which an adversary can exploit to make inferences about the training data points. These three datasets have shown moderate to severe vulnerability across the vulnerable classifiers and low vulnerability on the non-vulnerable classifiers. Body and Optdigits, which have well-balanced classes, have shown the least vulnerability (even the vulnerable classifiers have shown low vulnerability (<1.4) for these two datasets). The equitable distribution of points mitigates bias and reduces vulnerability, irrespective of the classifiers they train. 

Five of the ten classifiers were identified as vulnerable, and they cover all three categories of interpretability -- high interpretability (Decision Tree), moderate interpretability (k-Nearest Neighbor), and low interpretability (Random Forest, XGB, and MLP-deep). In our study, we could not establish any strong association between interpretability and vulnerability of classifiers. When we consider the aspects of linearity, non-linearity, and tree-based construct, we found the linear classifiers (SGD and Logistic Regression) to be not vulnerable. Three (Decision Tree, Random Forest, and XGB) out of four tree-based classifiers and two (kNN and MLP-deep) out of three non-linear classifiers are vulnerable.

\begin{figure*}[htbp]
   \hspace*{-2.8cm}  
    \includegraphics[scale=0.5]{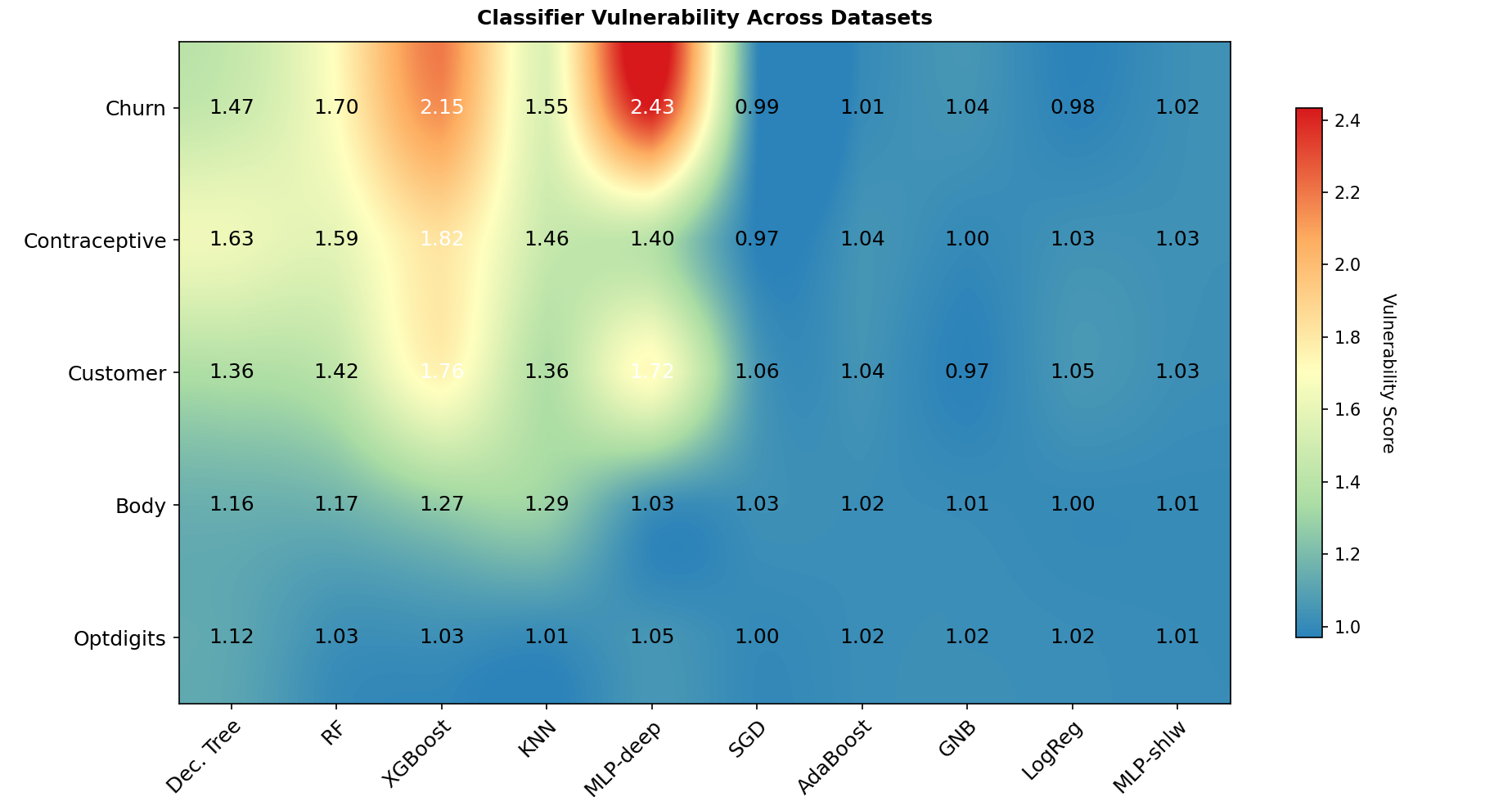}
    \caption{Vulnerability of classifiers on different datasets. The figure has a central component that shows the vulnerabilities, and the side panel shows the color legend. Each row of the main figure corresponds to a dataset, and each column corresponds to a classifier. The vulnerabilities of the imbalanced datasets are arranged in the top three rows (Churn, Contraceptive, and Customer), and the bottom two rows show the vulnerabilities of the balanced datasets. The vulnerability values of the vulnerable classifiers (Decision Tree, Random Forest, XGBoost, KNN, and MLP-deep) are reported on the left side of the table, and the non-vulnerable classifiers (SGD, AdaBoost, GNB, Logistic Regression, and MLP-shallow) are reported on the right side of the table.}
\end{figure*}

Our analysis shows that there are two factors which contribute to vulnerability of the classifiers -- i] the procedural framework of the classifiers, and ii] the distributional shift in classes of the datasets.

\subsubsection{Influence of number of layers and number of neurons on vulnerability of Multi Layer Perceptron (MLPClassifier)}
MLPClassifier's ability to model complex, non-linear relationships in data makes it a versatile choice for modern machine learning tasks.
Its architecture consists of multiple hidden layers and a varying number of neurons in each layer. Predictive power of an MLPClassifier is strongly dependent on these two parameters (along with other parameters like activation function and learning rates). We have investigated the association between vulnerability vs. number of neurons in layers (considering only one hidden layer) and vulnerability vs. number of layers. The number of neurons and the number of layers vary from 4 to 128, and 1 to 8, respectively. The investigation is carried out on all five datasets to ensure comprehensive evaluation.

The empirical results are illustrated in Figures 3-7. They show that the vulnerability of a classifier has a positive correlation with both the number of neurons and the number of layers for some of the datasets. The trend is evident in Churn, Contraceptive, and Customer. The variational shift and class imbalance of these datasets get embedded in the classifier, and it accentuates with increasing complexity of the classifier. However, this trend is absent in Body and Optdigits, and vulnerability shows zero correlation with the changing number of layers and neurons. These two datasets possess a balanced class distribution and are mostly devoid of distributional shift, which accounts for the uniformity in their training and test performances. 
\begin{figure*}[!h]
    \centering
    \subfloat[Vulnerability vs. number of layers\label{fig:layers}]{
        \includegraphics[width=\textwidth]{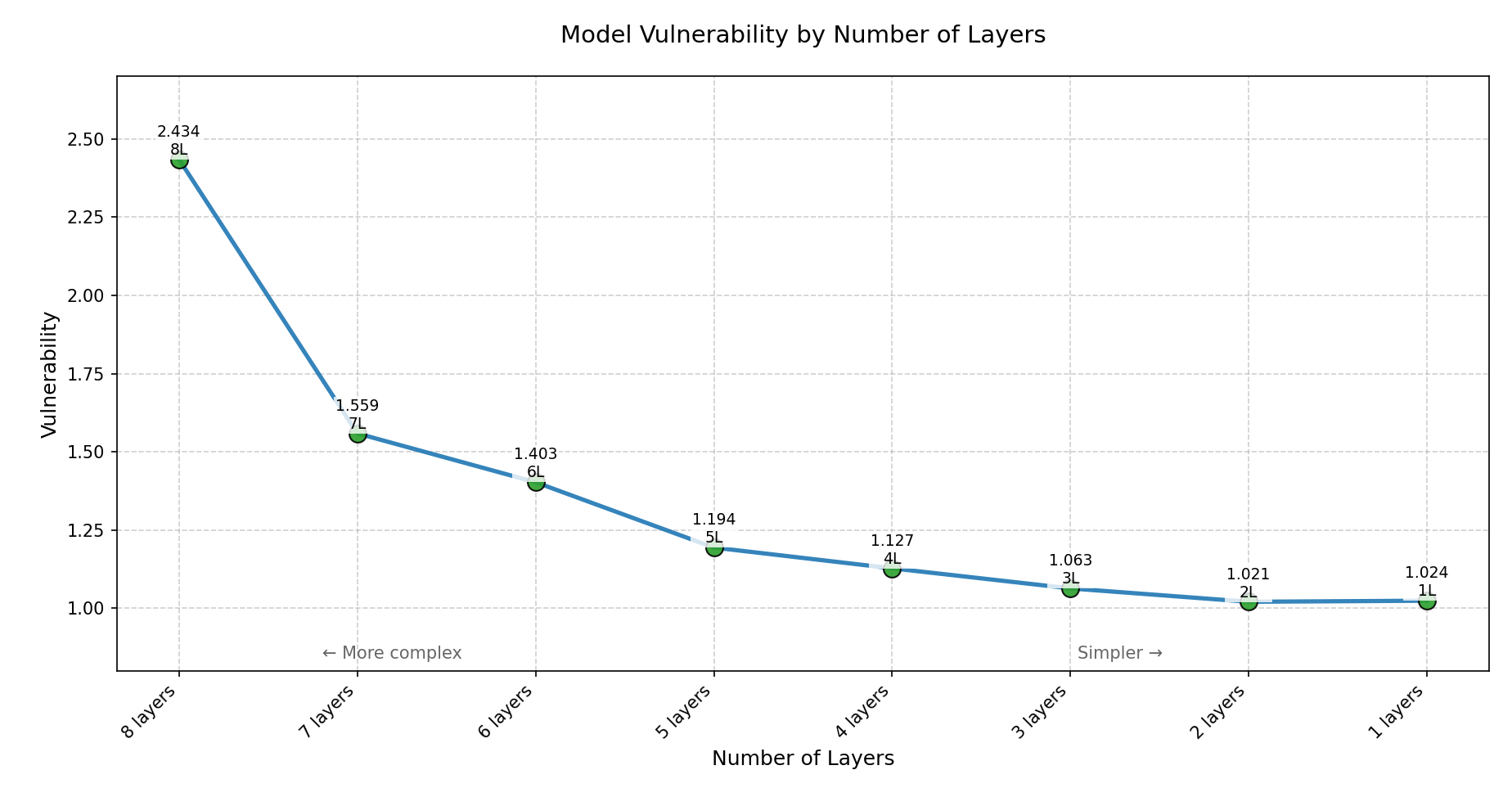}}
    
    \vspace{0.5cm} 
    
    \subfloat[Vulnerability vs. number of neurons\label{fig:neurons}]{
        \includegraphics[width=\textwidth]{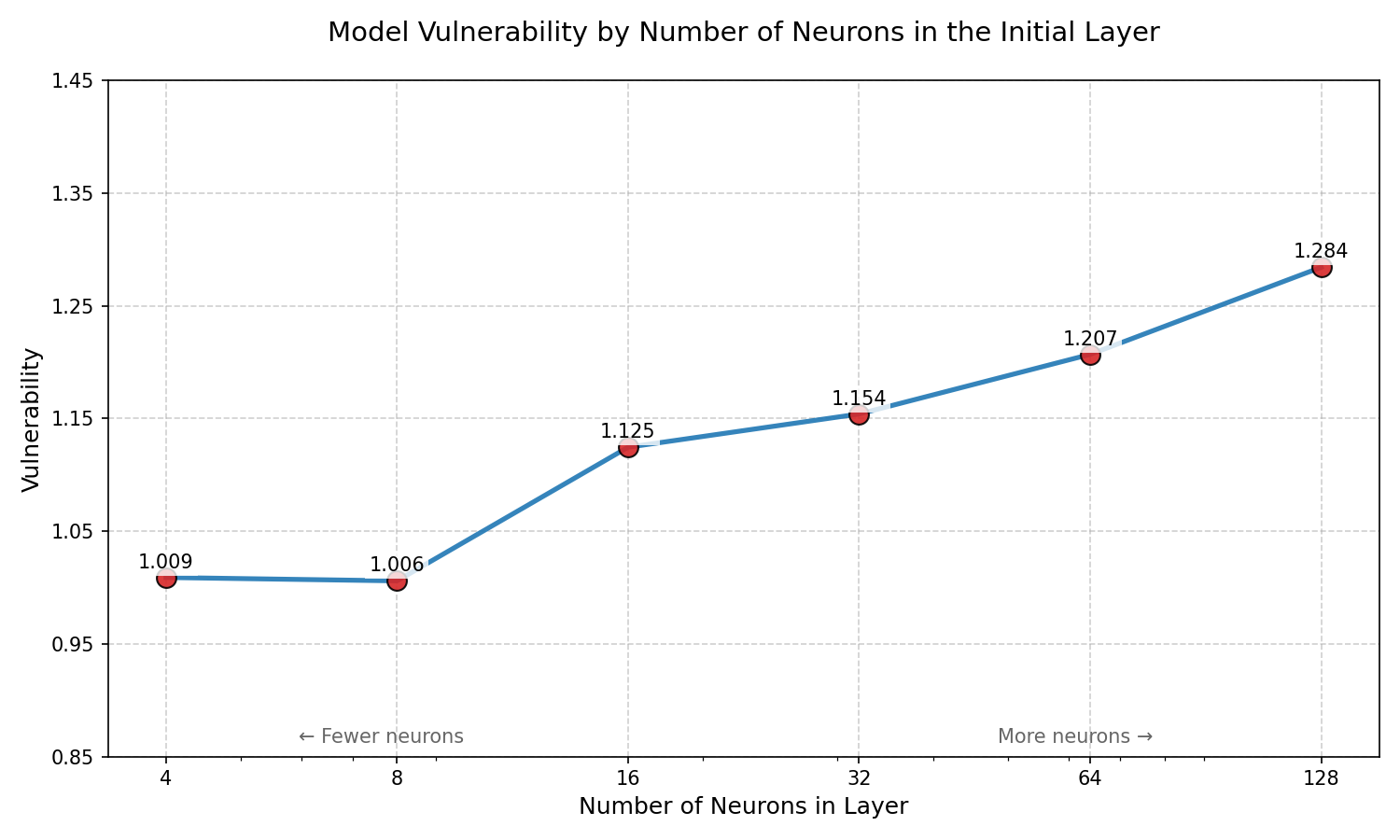}}
    
    \caption{Analysis of Churn dataset showing vulnerability as a function of (a) number of layers and (b) number of neurons. This dataset has a substantial amount of imbalance between its two classes. Bias in the data couples with the increasing number of neurons and the increasing number of layers, and contributes to a positive correlation of either with vulnerability.}
    \label{fig:churn_analysis}
\end{figure*}

\begin{figure*}[h]
    \centering
    \subfloat[Vulnerability vs. number of layers\label{fig:layers}]{
        \includegraphics[width=\textwidth]{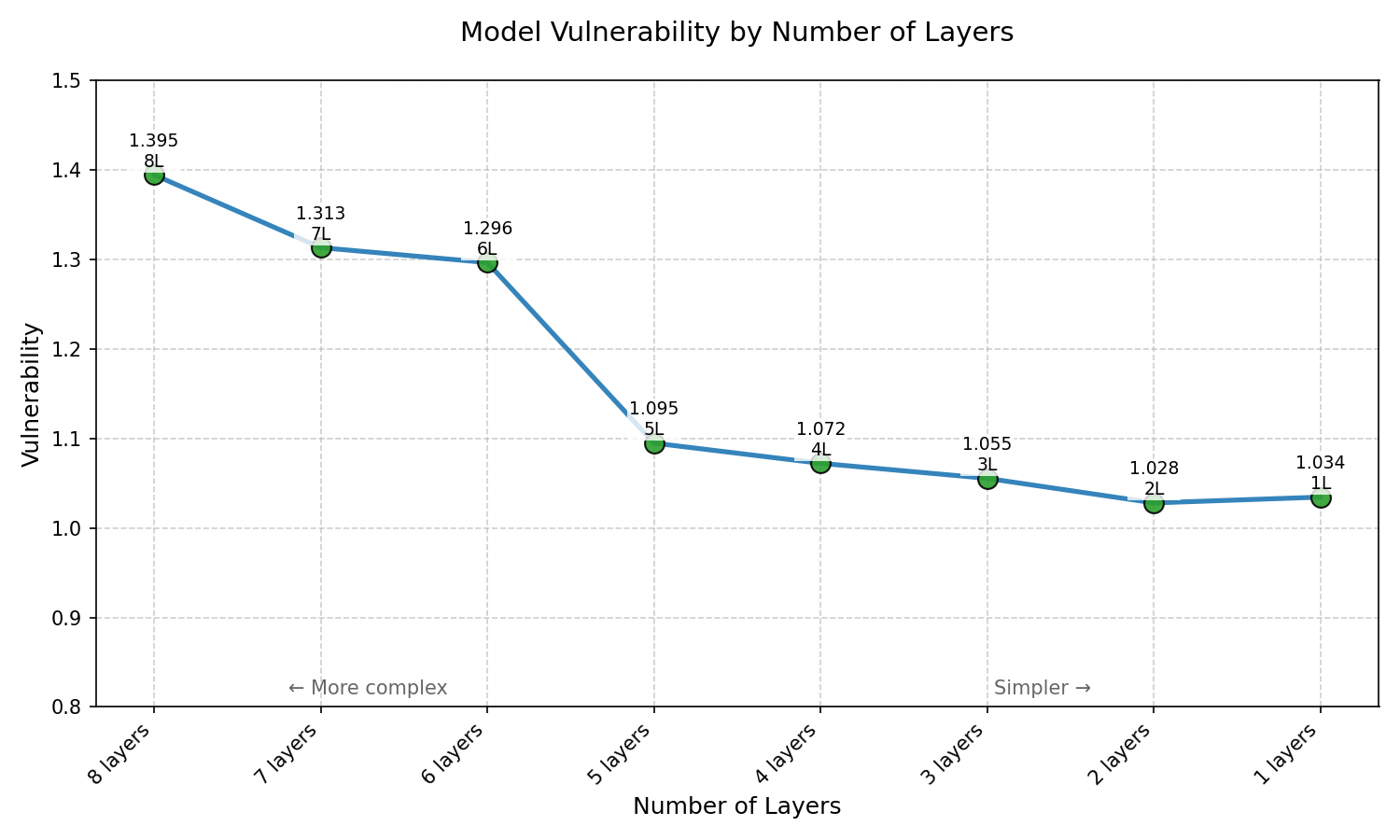}
    }
    
    \vspace{0.5cm} 
    
    \subfloat[Vulnerability vs. number of neurons\label{fig:neurons}]{
        \includegraphics[width=\textwidth]{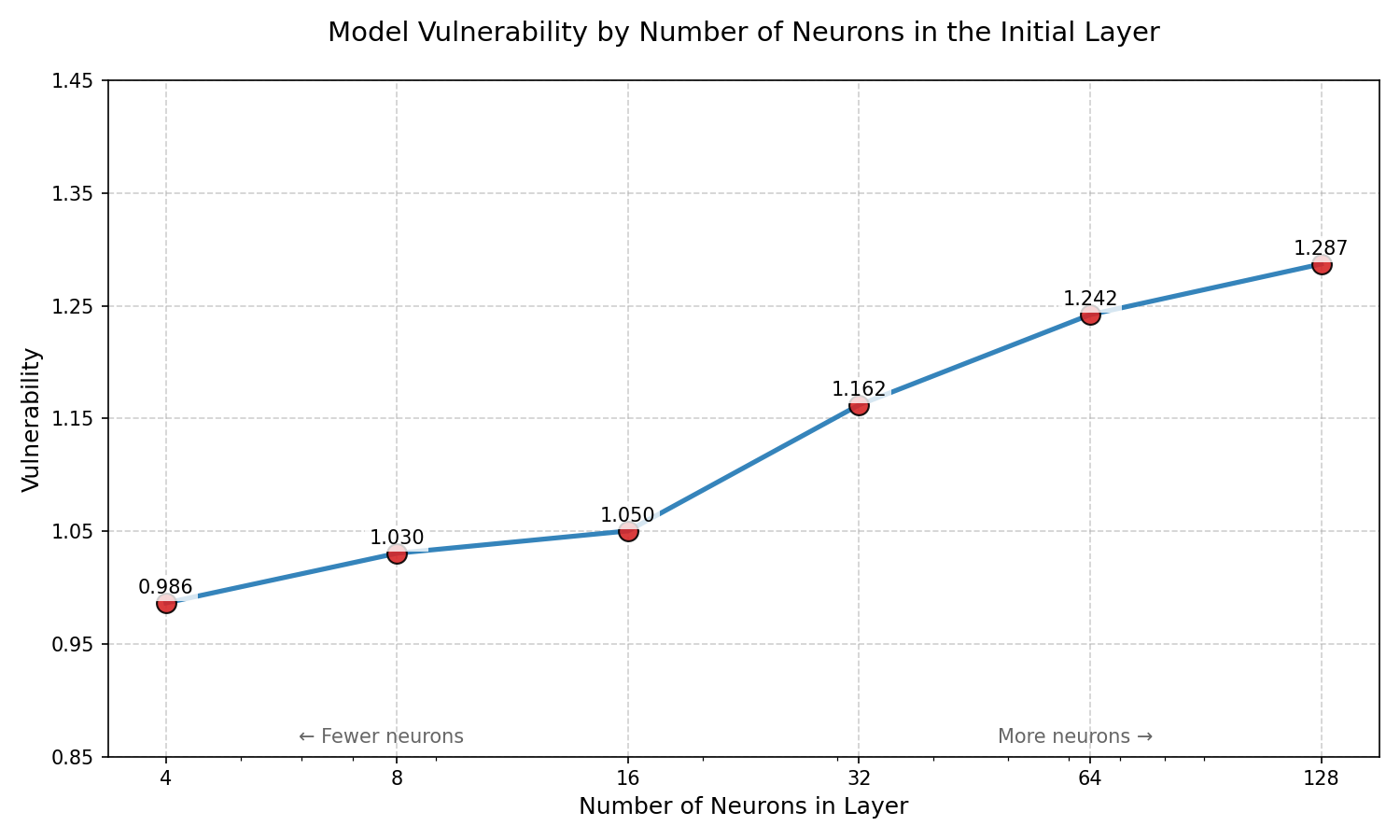}
    }
    
    \caption{Analysis of Contraceptive dataset showing vulnerability as a function of (a) number of layers and (b) number of neurons. This dataset has an admissible degree of imbalance between its classes. Bias in the data, coupled with an increasing number of neurons and an increasing number of layers, contributes to an increase in vulnerability.}
    \label{fig:churn_analysis}
\end{figure*}

\begin{figure*}[h]
    \centering
    \subfloat[Vulnerability vs. number of layers\label{fig:layers}]{
        \includegraphics[width=\textwidth]{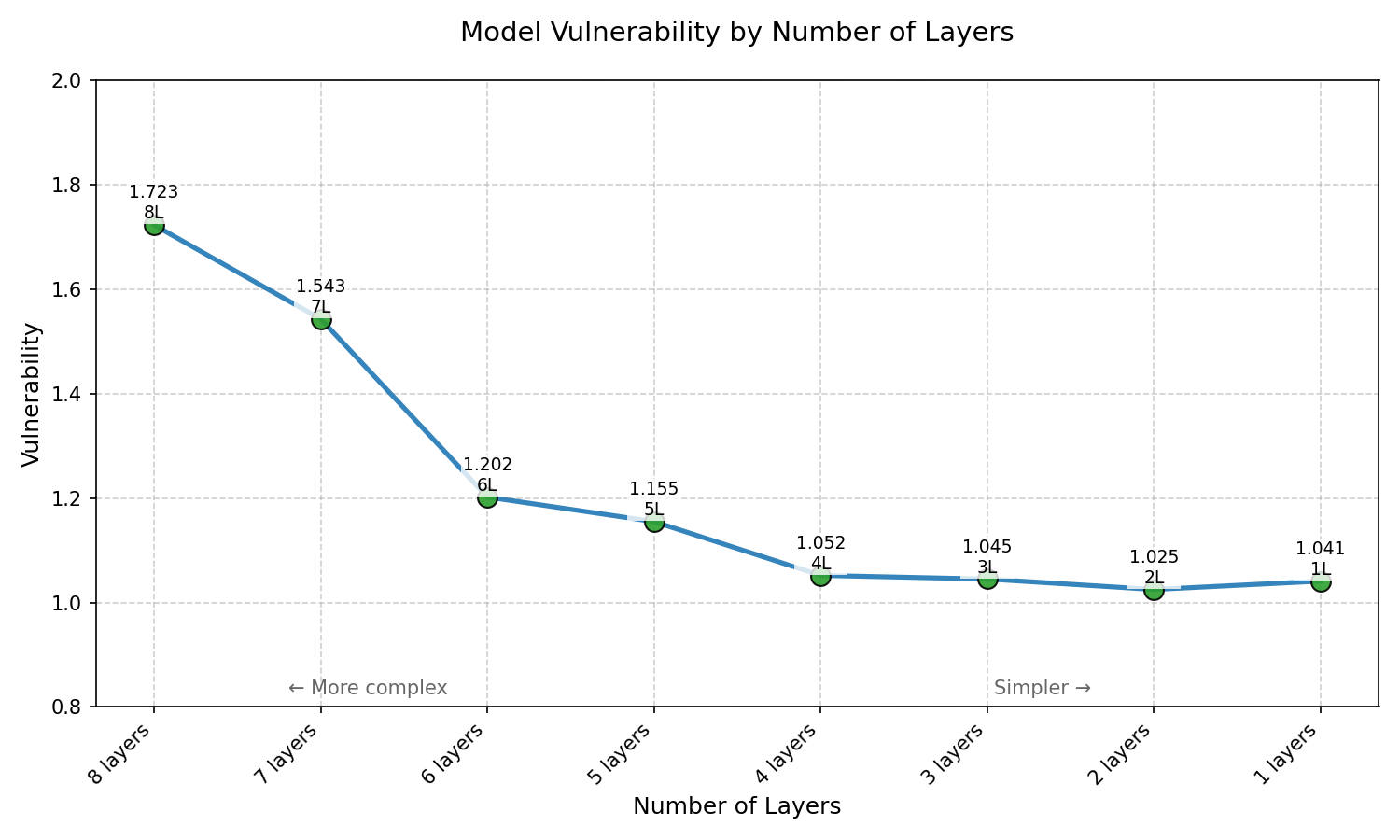}
    }
    
    \vspace{0.5cm} 
    
    \subfloat[Vulnerability vs. number of neurons\label{fig:neurons}]{
        \includegraphics[width=\textwidth]{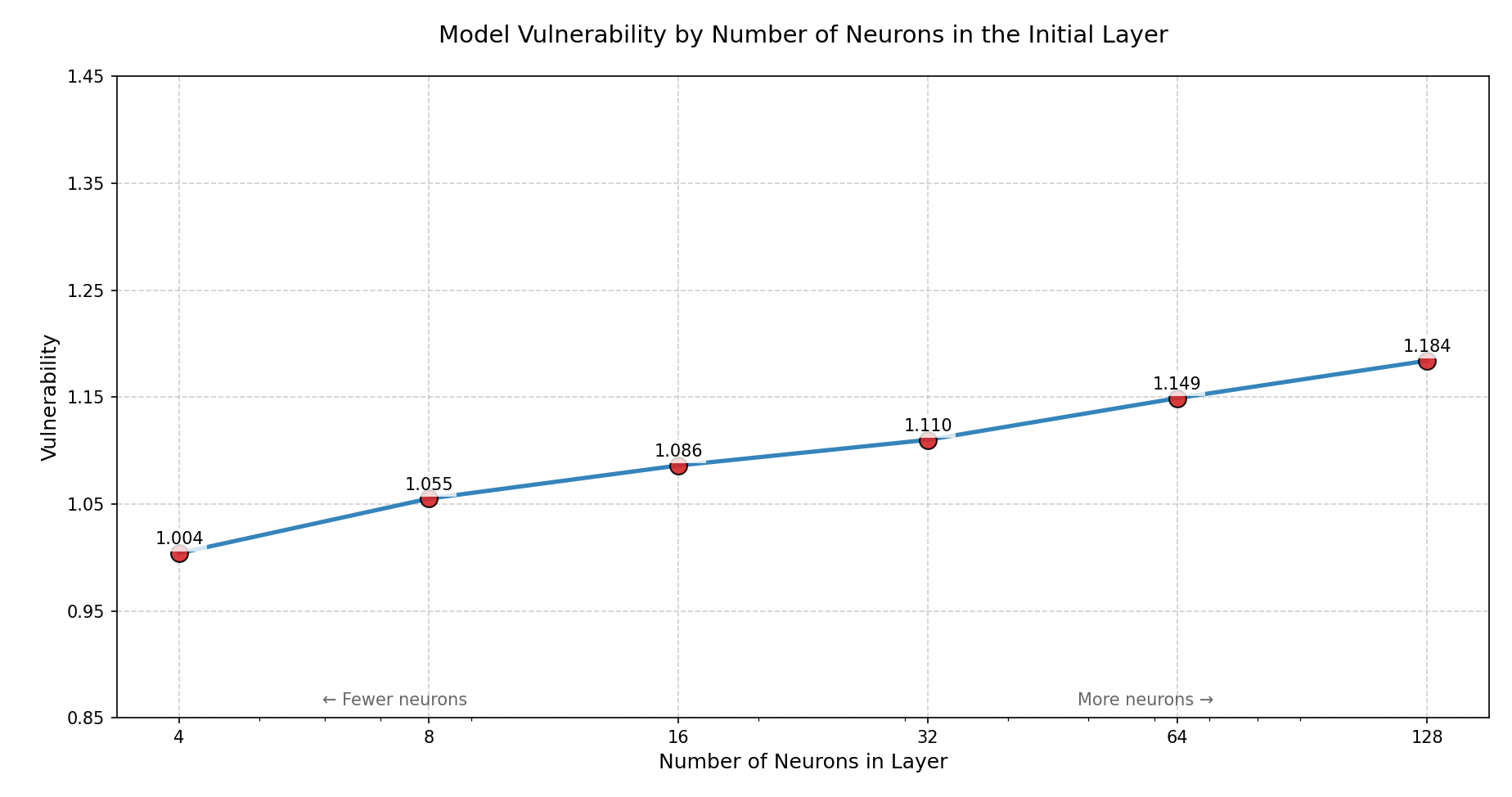}
    }
    
    \caption{Analysis of Customer dataset showing vulnerability as a function of (a) number of layers and (b) number of neurons. This dataset has a substantial amount of imbalance between its three classes. Bias in the data couples with the increasing number of neurons and the increasing number of layers, and renders a positive correlation of either with vulnerability.}
    \label{fig:churn_analysis}
\end{figure*}

\begin{figure*}[h]
    \centering
    \subfloat[Vulnerability vs. number of layers\label{fig:layers}]{
        \includegraphics[width=\textwidth]{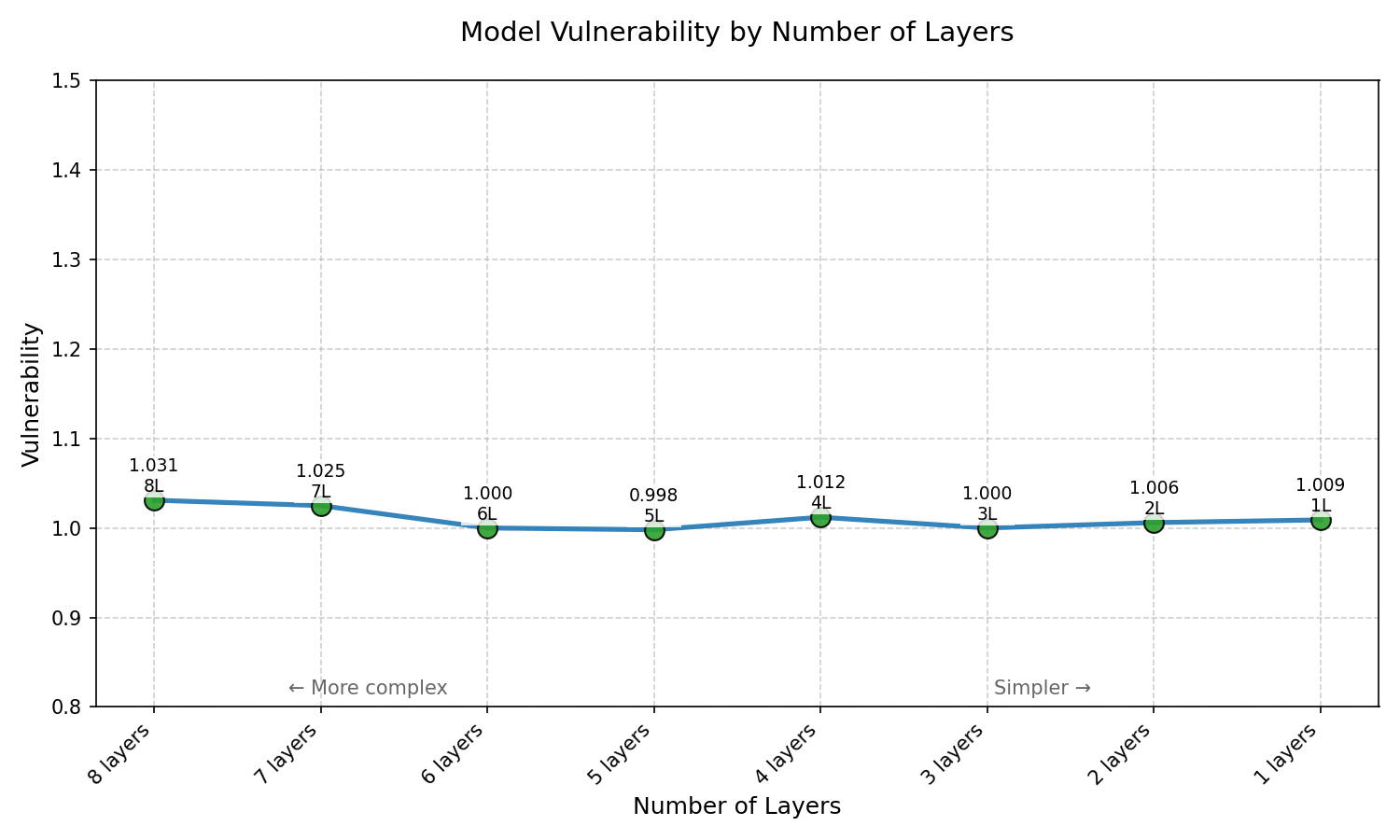}
    }
    
    \vspace{0.5cm} 
    
    \subfloat[Vulnerability vs. number of neurons\label{fig:neurons}]{
        \includegraphics[width=\textwidth]{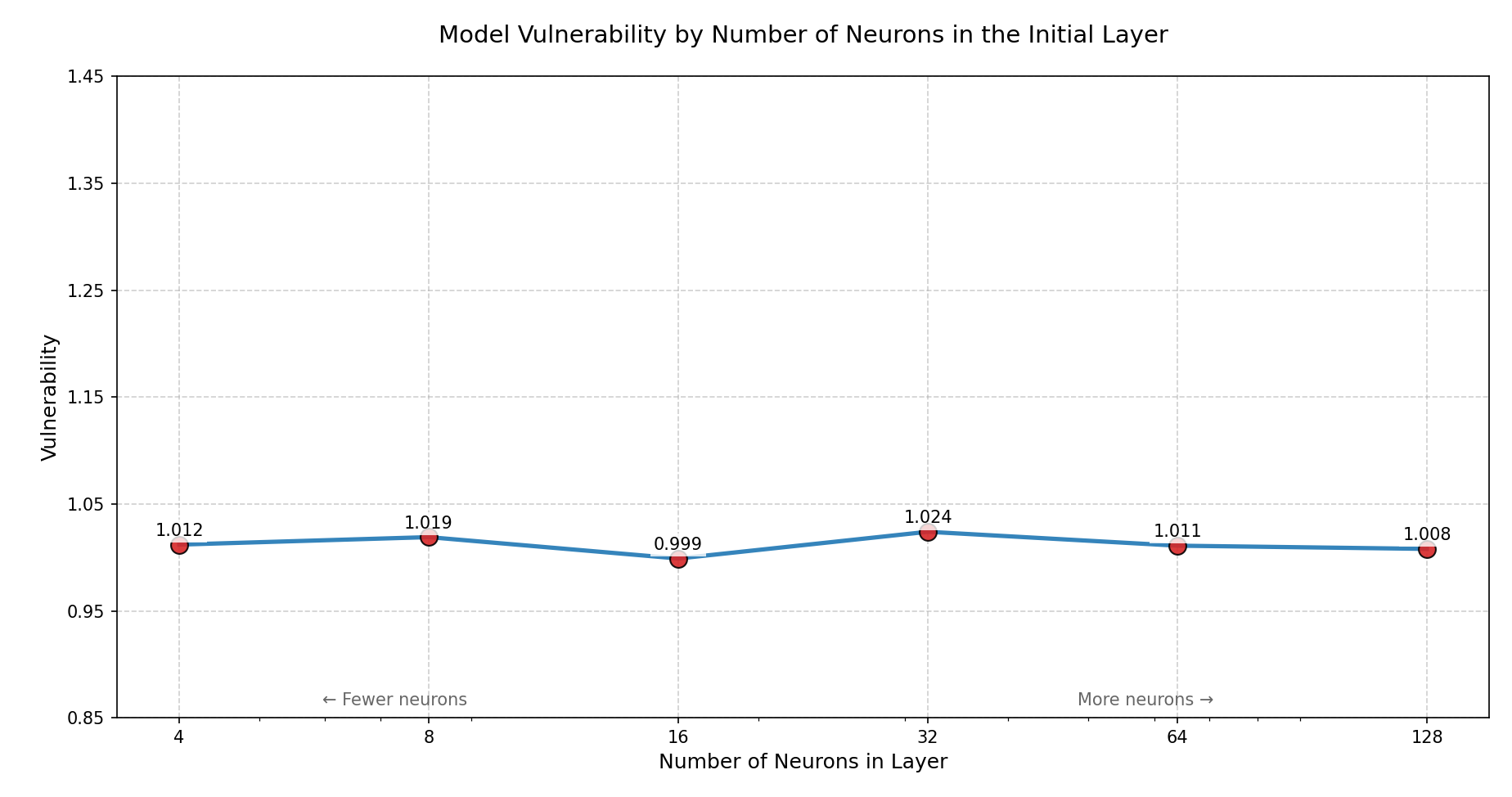}
    }
    
    \caption{Analysis of Body dataset showing vulnerability as a function of (a) number of layers and (b) number of neurons. This dataset is devoid of any class bias, rendering it a non-vulnerable classifier. The vulnerability of the classifier trained on this data has zero correlation with both the number of layers and the number of neurons.}
    \label{fig:churn_analysis}
\end{figure*}

\begin{figure*}[h]
    \centering
    \subfloat[Vulnerability vs. number of layers\label{fig:layers}]{
        \includegraphics[width=\textwidth]{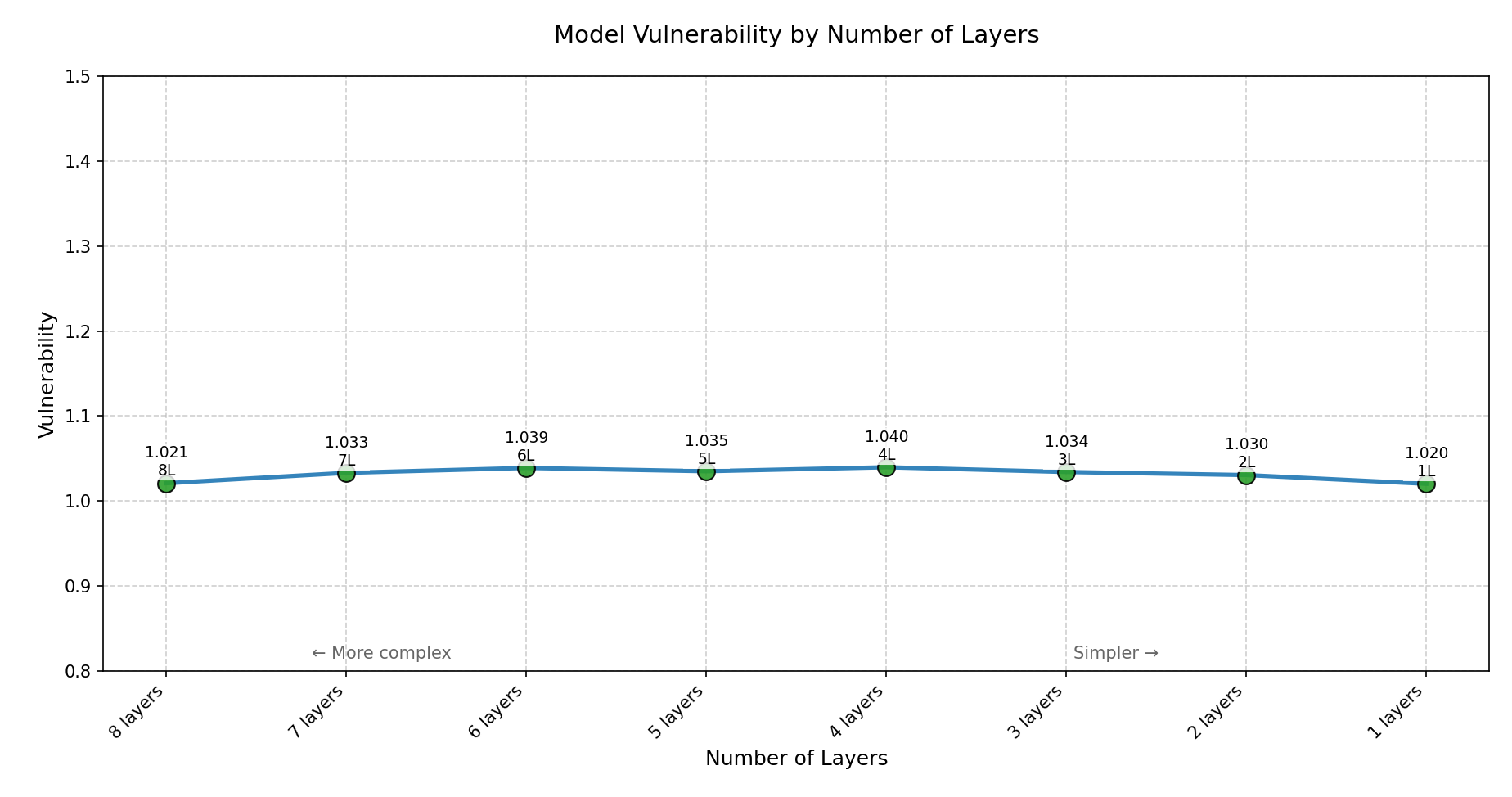}}

    \vspace{0.5cm} 
    
    \subfloat[Vulnerability vs. number of neurons\label{fig:neurons}]{
        \includegraphics[width=\textwidth]{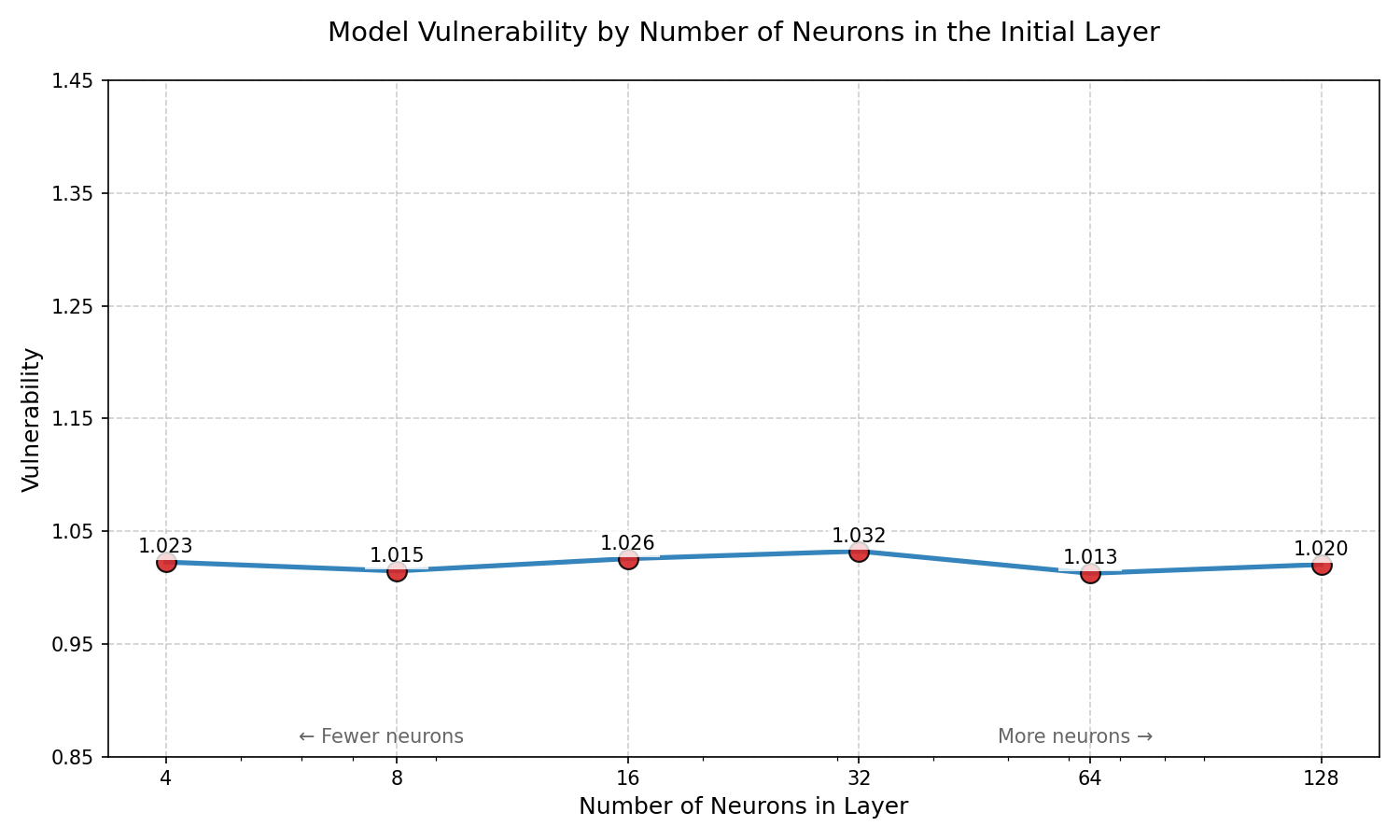}}

    \caption{Analysis of Optdigits dataset showing vulnerability as a function of (a) number of layers and (b) number of neurons. This dataset is devoid of any class bias, rendering it a non-vulnerable classifier. The vulnerability of the classifier trained on this data has zero correlation with both the number of layers and the number of neurons.}
    \label{fig:churn_analysis}
\end{figure*}

\FloatBarrier
\begin{table*}[h!]
\caption{Change of classifier's vulnerability with incorporation of data obfuscation for \emph{Churn} dataset}
\begin{tabular}{|l|r|r r|r r|}
\hline
\multicolumn{1}{|c|}{}&Original&\multicolumn{2}{|c|}{LSH}&\multicolumn{2}{|c|}{Hamming} \\
\hline
Datasets&$vul$&${vul}_{obf}$&${vul}_{change}$&${vul}_{obf}$&${vul}_{change}$ \\
\hline
Decision Tree&1.48&1.02&31.19$\%$&1.03&29.95$\%$\\

Random Forest&1.37&0.99&27.54$\%$&1.02&25.72$\%$\\

XGB&2.11&1.01&52.13$\%$&1.01&52.37$\%$\\

kNN&1.56&0.94&39.73$\%$&0.97&37.94$\%$\\

MLP-deep&2.43&1.11&54.32$\%$&1.08&55.56$\%$\\

SGD&1.04&1.01&2.86$\%$&0.97&6.59$\%$\\

Adaboost&0.97&0.95&2.09$\%$&0.99&-1.34$\%$\\

GNB&1.05&1.02&2.71$\%$&1.03&1.81$\%$\\

Log Reg&0.93&0.95&-2.01$\%$&1.02&-9.68$\%$\\

MLP-shallow&1.02&1.00&2.00$\%$&1.03&-0.33$\%$\\

\hline

\hline
\end{tabular}
\end{table*}
\FloatBarrier

\FloatBarrier
\begin{table*}[h!]
\caption{Change of classifier's vulnerability with incorporation of data obfuscation for \emph{Contraceptive} dataset}
\begin{tabular}{|l|r|r r|r r|}
\hline
\multicolumn{1}{|c|}{}&Original&\multicolumn{2}{|c|}{LSH}&\multicolumn{2}{|c|}{Hamming} \\
\hline
Datasets&$vul$&${vul}_{obf}$&${vul}_{change}$&${vul}_{obf}$&${vul}_{change}$ \\
\hline
Decision Tree&1.63&1.15&29.45$\%$&1.07&34.36$\%$\\

Random Forest&1.59&1.04&34.59$\%$&1.06&33.33$\%$\\

XGB&1.82&1.02&43.96$\%$&1.01&44.51$\%$\\

kNN&1.46&0.96&34.24$\%$&0.96&34.24$\%$\\

MLP-deep&1.40&0.98&30.00$\%$&1.03&26.43$\%$\\

SGD&0.97&1.00&-3.09$\%$&0.98&-1.03$\%$\\

Adaboost&1.04&0.97&6.73$\%$&1.00&4.00$\%$\\

GNB&1.00&0.98&2.00$\%$&1.03&-3.00$\%$\\

Log Reg&1.03&0.98&4.85$\%$&1.01&-1.94$\%$\\

MLP-shallow&1.03&1.05&-1.94$\%$&1.04&-0.97$\%$\\

\hline

\hline
\end{tabular}
\end{table*}
\FloatBarrier

\FloatBarrier
\begin{table*}[h!]
\caption{Change of classifier's vulnerability with incorporation of data obfuscation for \emph{Customer} dataset}
\begin{tabular}{|l|r|r r|r r|}
\hline
\multicolumn{1}{|c|}{}&Original&\multicolumn{2}{|c|}{LSH}&\multicolumn{2}{|c|}{Hamming} \\
\hline
Datasets&$vul$&${vul}_{obf}$&${vul}_{change}$&${vul}_{obf}$&${vul}_{change}$ \\
\hline
Decision Tree&1.37&1.02&25.53$\%$&1.01&25.57$\%$\\

Random Forest&1.42&1.17&30.00$\%$&1.03&27.46$\%$\\

XGB&1.76&1.00&43.07$\%$&1.07&39.36$\%$\\

kNN&1.36&1.00&26.75$\%$&1.03&24.71$\%$\\

MLP-deep&1.72&1.00&41.86$\%$&0.99&42.44$\%$\\

SGD&1.06&0.99&6.92$\%$&1.00&5.40$\%$\\

Adaboost&1.04&1.00&4.44$\%$&0.93&10.57$\%$\\

GNB&0.97&1.08&-11.23$\%$&0.94&3.10$\%$\\

Log Reg&1.05&0.97&7.64$\%$&1.04&1.21$\%$\\

MLP-shallow&1.03&1.02&00.97$\%$&1.01&1.94$\%$\\
\hline

\hline
\end{tabular}
\end{table*}
\FloatBarrier

\FloatBarrier
\begin{table*}[h!]
\caption{Change of classifier's vulnerability with incorporation of data obfuscation for \emph{Body} dataset}
\begin{tabular}{|l|r|r r|r r|}
\hline
\multicolumn{1}{|c|}{}&Original&\multicolumn{2}{|c|}{LSH}&\multicolumn{2}{|c|}{Hamming} \\
\hline
Datasets&$vul$&${vul}_{obf}$&${vul}_{change}$&${vul}_{obf}$&${vul}_{change}$ \\
\hline
Decision Tree&1.16&1.04&10.71$\%$&1.03&11.21$\%$\\

Random Forest&1.17&0.97&16.79$\%$&1.01&13.64$\%$\\

XGB&1.17&1.00&14.59$\%$&1.07&5.98$\%$\\

kNN&1.30&1.00&23.07$\%$&0.96&25.62$\%$\\

MLP-deep&1.03&1.01&1.94$\%$&1.00&42.91$\%$\\

SGD&1.03&1.01&1.69$\%$&0.98&4.45$\%$\\

Adaboost&1.02&1.00&1.66$\%$&1.04&-1.78$\%$\\

GNB&1.01&1.04&-2.77$\%$&1.01&0.51$\%$\\

Log Reg&1.00&1.05&-4.30$\%$&1.09&-7.53$\%$\\

MLP-shallow&1.01&1.02&-00.99$\%$&1.01&0.00$\%$\\
\hline

\hline
\end{tabular}
\end{table*}
\FloatBarrier

\FloatBarrier
\begin{table*}[h!]
\caption{Change of classifier's vulnerability with incorporation of data obfuscation for \emph{Optdigits} dataset}
\begin{tabular}{|l|r|r r|r r|}
\hline
\multicolumn{1}{|c|}{}&Original&\multicolumn{2}{|c|}{LSH}&\multicolumn{2}{|c|}{Hamming} \\
\hline
Datasets&$vul$&${vul}_{obf}$&${vul}_{change}$&${vul}_{obf}$&${vul}_{change}$ \\
\hline
Decision Tree&1.12&1.02&8.93$\%$&1.03&8.04$\%$\\

Random Forest&1.03&0.95&7.77$\%$&0.97&5.83$\%$\\

XGB&1.03&0.99&3.88$\%$&1.07&6.79$\%$\\

kNN&1.01&1.03&-1.98$\%$&1.02&-0.99$\%$\\

MLP-deep&1.05&1.03&1.05$\%$&0.98&6.67$\%$\\

SGD&1.00&0.97&-3.00$\%$&0.98&-2.00$\%$\\

Adaboost&1.02&1.01&0.98$\%$&1.02&0.00$\%$\\

GNB&1.02&1.05&-2.94$\%$&1.04&-1.96$\%$\\

Log Reg&1.02&1.05&-2.94$\%$&1.05&-2.94$\%$\\

MLP-shallow&1.00&1.02&-2.00$\%$&0.97&3.00$\%$\\
\hline

\hline
\end{tabular}
\end{table*}
\FloatBarrier

\subsection{Experiment 2:} Tables 1-5 report the change of vulnerability on adopting data obfuscation protocols. Tables 1, 2, 3, 4, and 5 are dedicated to Churn, Customer, and Body. For each classifier-dataset pair, five values are reported: vulnerability on the original dataset, vulnerability on LSH-encoding-based obfuscated data, the corresponding change with respect to the original dataset, vulnerability on Hamming-encoding-based obfuscated data, and the corresponding change with respect to the original dataset. The range of ${vul}_{change}$ is $11.21\%$ to $52.37\%$ for vulnerable classifiers. The results demonstrate the two obfuscation schemas' ability to remove the training dataset's footprints from the classifiers. Both methods have rendered vulnerability scores of approximately 1 for all classifier-dataset pairs. A vulnerability score equal to 1 indicates equal performance (accuracy or average precision) on the training and test data partition.

The vulnerability change for the non-vulnerable classifiers is $-9.68\%$ to $10.57\%$. Note that the incorporation of obfuscation schemes in non-vulnerable classifiers (SGD, Adaboost, Gaussian Naive Bayes, Logistic Regression, and MLP-shallow) has rendered negative ${vul}_{change}$ values in some cases. These results show that data obfuscation is inessential for non-vulnerable classifiers.

\FloatBarrier
\begin{table*}[h!]
\caption{Privacy-performance tradeoff obtained for \emph{Churn} dataset on the vulnerable classifiers.}
\begin{tabular}{|l|r|r|}
\hline
Datasets&\multicolumn{2}{|c|}{privacy-performance trade-off}\\
\hline
&LSH encoding&Hamming encoding \\
\hline
Decision Tree&1.13&1.13\\

Random Forest&0.94&0.92\\

XGB&1.14&1.28\\

kNN&1.30&1.29\\

MLP-deep&1.19&1.16\\




\hline

\hline
\end{tabular}
\end{table*}

\FloatBarrier

\FloatBarrier
\begin{table*}[h!]
\caption{Privacy-performance tradeoff obtained for \emph{Contraceptive} dataset on the vulnerable classifiers.}
\begin{tabular}{|l|r|r|}
\hline
Datasets&\multicolumn{2}{|c|}{privacy-performance trade-off}\\
\hline
&LSH encoding&Hamming encoding \\
\hline
Decision Tree&1.08&1.10\\

Random Forest&0.97&0.95\\

XGB&1.18&1.21\\

kNN&1.33&1.23\\

MLP-deep&1.22&1.20\\




\hline

\hline
\end{tabular}
\end{table*}

\FloatBarrier

\FloatBarrier
\begin{table*}[h!]
\caption{Privacy-performance tradeoff obtained for \emph{Customer} dataset on the vulnerable classifiers.}
\begin{tabular}{|l|r|r|}
\hline
Datasets&\multicolumn{2}{|c|}{privacy-performance trade-off}\\
\hline
&LSH encoding&Hamming encoding \\
\hline
Decision Tree&1.20&1.19\\

Random Forest&1.17&1.11\\

XGB&1.46&1.30\\

kNN&1.35&1.26\\

MLP-deep&1.27&1.32\\




\hline

\hline
\end{tabular}
\end{table*}

\FloatBarrier

\FloatBarrier
\begin{table*}[h!]
\caption{Privacy-performance tradeoff obtained for \emph{Body} dataset on the vulnerable classifiers.}
\begin{tabular}{|l|r|r|}
\hline
Datasets&\multicolumn{2}{|c|}{privacy-performance trade-off}\\
\hline
&LSH encoding&Hamming encoding \\
\hline
Decision Tree&0.80&0.42\\

Random Forest&0.79&0.38\\

XGB&0.78&0.38\\

kNN&0.61&0.65\\

MLP-deep&0.85&0.62\\




\hline

\hline
\end{tabular}
\end{table*}
\FloatBarrier

\FloatBarrier
\begin{table*}[h!]
\caption{Privacy-performance tradeoff obtained for \emph{Optdigits} dataset on the vulnerable classifiers.}
\begin{tabular}{|l|r|r|}
\hline
Datasets&\multicolumn{2}{|c|}{privacy-performance trade-off}\\
\hline
&LSH encoding&Hamming encoding \\
\hline
Decision Tree&0.76&0.40\\

Random Forest&0.78&0.39\\

XGB&0.74&0.41\\

kNN&0.65&0.47\\

MLP-deep&0.80&0.64\\




\hline

\hline
\end{tabular}
\end{table*}
\FloatBarrier

\subsection{Experiment 3:} Tables 6-10 report the privacy-performance trade-off achieved by incorporating data obfuscation. We have used the scores (privacy and performance) on original datasets as the baseline. Hence, following Equation 13, the baseline is one. A value higher than 1 denotes the feasible admissibility of an obfuscation protocol while curtailing the footprints of training data. Data from Table 6 shows that on \emph{Churn}, LSH-encoding and Hamming's encoding have rendered privacy-performance trade-off value $>1$ for Decision Tree, XGB, kNN, and MLP-deep classifiers. For Random Forest, the value is slightly less than one in both cases. On \emph{Contraceptive} dataset, both the techniques (LSH and Hamming) have rendered a privacy-performance trade-off close to 1 or above for all the classifiers. For the \emph{Customer} dataset, both obfuscation schemes have rendered privacy-performance trade-off value greater than one on all classifiers.

{The scenario on \emph{Body} and \emph{Optdigits} datasets is slightly different from the remaining three datasets. The obfuscation protocols have failed to render an admissible privacy-performance trade-off across all the classifiers. LSH encoding has rendered scores around 0.8 for Decision Tree, Random Forest, XGB, and MLP-deep classifiers. On the kNN classifier, the score is further lower at 0.65. On Hamming encoding, the scores range from 0.38 to 0.63. These scores show that despite improving upon footprint mitigation, data obfuscation is not a practical solution for this dataset. The balanced nature of the classes in these datasets is to be accounted for. The uniformity of the balanced classes gets further randomized through obfuscation, and a sub-optimal classifier is trained. These classifiers produce outputs with significant randomness and low veracity across both training and test partitions. This leads to a substantial fall in privacy-performance trade-off values.

\section{Discussion}
In this section, we present the key findings of our study and their implications for Information Systems research and business applications. Additionally, we outline the study's limitations and highlight open challenges that warrant further investigation in future research.

\subsection{Key findings of this study}
This study identifies and conceptualizes passive security and privacy vulnerabilities inherent in predictive models. Our analysis highlights that these vulnerabilities are not merely theoretical concerns but practical issues that, if left unaddressed, could lead to severe data breaches, conflicting with Article 17 of the GDPR. Notably, our findings indicate that classifier vulnerability arises from two key factors: the training/methodology of classifiers and distributional shifts in data classes. Quantitative and qualitative shifts in data distributions introduce bias, which becomes embedded in the classifier, increasing the likelihood of training data exposure.

To mitigate this risk, we investigate data obfuscation, a common technique for preserving data privacy. Our empirical findings confirm its effectiveness in reducing classifier vulnerability. Furthermore, our study on the privacy-performance trade-off demonstrates that data obfuscation techniques can be employed in real-world applications where both computational efficiency and data privacy are critical.

\subsection{Implications for IS Research and Business}
User and data privacy are crucial considerations for businesses. Existing Information Systems (IS) research primarily addresses privacy breaches resulting from active intrusions and deliberate attacks. In contrast, this study makes several key contributions by addressing passive privacy threats:\\

\noindent \emph{Conceptualizing Classifier Vulnerabilities:} 
This research highlights how training data can be inherently embedded in predictive models, leading to unintended privacy risks. Such vulnerabilities conflict with Article 17 of the GDPR, making non-compliance potentially costly for businesses. In addition to conceptualizing vulnerability, we introduce a framework for its quantification, allowing practitioners to assess and compare classifiers based on their susceptibility to data leakage. This framework assists: (a)Users and practitioners must identify vulnerable classifiers to make informed deployment choices in sensitive applications; (b) Model designers in understanding and mitigating the root causes of classifier vulnerability.

\noindent \emph{Proposing a Remedial Framework:} We explore a remedial strategy based on data obfuscation, testing its effectiveness through controlled experiments. The empirical results offer multiple options for mitigating classifier vulnerability: (a) Selecting a classifier that inherently minimizes data exposure; (b) Adjusting the training process to reduce vulnerability; and (c) Applying data obfuscation techniques to enhance privacy protection.

\noindent \emph{Balancing Privacy and Performance:} While privacy is paramount, businesses must also maintain model performance. Our introduction of a privacy-performance trade-off index enables practitioners to assess the feasibility of data obfuscation techniques, helping them balance privacy preservation with model effectiveness before deployment. To be usable in business ventures, the data obfuscation protocols have to be minimally corruptive to the systems' performances. The quantification of \emph{privacy-performance trade-off} allows the integration of the reduction in vulnerability and loss in performance under a single head. It enables the users to decide whether to adopt data obfuscation protocols before deploying the models.\\

\subsection{Limitations and Future Research Directions}
Despite its contributions, this study has certain limitations that warrant further investigation:\\

\noindent {\emph{Exploring Additional Forms of Passive Vulnerability:}}
This study focuses on classifier vulnerability arising from performance differences between training and test data \cite{bozorgi2010beyond}. Future research should examine cases where classifiers produce similar performance metrics across both datasets, which may introduce alternative privacy risks.

\noindent {\emph{Enhancing Performance with Privacy Measures:}
While data obfuscation reduces vulnerability, it can also degrade model performance in some cases \cite{melzi2024overview}. Future studies should explore complementary techniques—such as adversarial training, differential privacy, or hybrid obfuscation strategies—to maintain model accuracy while preserving privacy.}\\

\noindent By addressing these open challenges, future research can further advance privacy-preserving methodologies in AI-driven decision-making systems.

\section{Conclusion}
This research investigates and theorizes a fundamental issue in classification models—footprints of training data—when applied in decision-making and policy implementation. This concern is particularly critical for the Information Systems (IS) community, as it directly challenges the Right to Erasure under GDPR. The findings of this study confirm that data footprints are embedded in a substantial class of classifiers, raising serious concerns about privacy and security in AI-driven decision systems.

Our theoretical and empirical analyses demonstrate that data imbalance and distributional shifts contribute significantly to classifier vulnerability. Bias introduced through these factors increases the risk of training data exposure, making it imperative for IS researchers and practitioners to approach classifier selection and training methodologies with caution. Addressing these vulnerabilities requires an informed balance between accuracy and privacy preservation. The quantification of vulnerability and the introduction of a \emph{privacy-performance trade-off framework} provide actionable insights for both researchers and practitioners. This framework enables users to assess classifier risks effectively, aiding in the deployment of models that align with privacy regulations while maintaining decision-making efficacy.

For the IS community, these findings underscore the necessity of integrating privacy-aware learning frameworks into AI-based decision-making systems. Future research should explore novel techniques for mitigating data footprints, such as adversarial privacy preservation, differential privacy, and fairness-aware machine learning \cite{ferrara2024fairness}. From a business perspective, organizations deploying AI models must carefully evaluate the trade-offs between privacy risks and model accuracy, ensuring compliance with legal frameworks while maintaining competitive performance. By adopting responsible AI practices, businesses can safeguard user privacy without compromising the reliability and effectiveness of predictive models \cite{deshpande2022responsible}.

This study contributes to the ongoing IS field discourse by offering theoretical foundations and empirical evidence on the critical intersection of privacy, accuracy, and vulnerability in classification models. The call for stronger privacy-preserving AI methodologies is a legal imperative and a business necessity in an increasingly data-driven world.

\section{Declarations}
\begin{itemize}
   
\item Ethics approval and consent to participate: No human or living subjects are used in the study.
 \item Consent for publication: Given by both authors,
\item Availability of data and material: Data will be made available on request.
\item Funding: None         
\item Authors' contributions: Payel Sadhukhan -- conceptualization, experiments, and writing, Tanujit Chakraborty -- Writing, editing, and refinement.
\item {Acknowledgements: We acknowledge the editors and learned reviewers for their insights and suggestions in improving the quality of this manuscript.}
\item Competing Interests: None.
\end{itemize}


\bibliographystyle{splncs04}
\bibliography{privacy-ref}

\end{document}